\documentclass[showpacs,amsmath,amssymb,twocolumn,superscriptaddress,notitlepage,preprintnumbers,pra]{revtex4-2}

\usepackage{graphicx}% Include figure files
\usepackage{dcolumn}% Align table columns on decimal point
\usepackage{bm}% bold math
\usepackage{booktabs}% for \toprule, \midrule, and \bottomrule
\usepackage{amsthm}
\usepackage{hyperref}
\usepackage{lineno}
\hypersetup{colorlinks=true, linkcolor=blue, citecolor=blue, urlcolor=black }

\newtheorem{theorem}{Theorem}[section]
\newtheorem{corollary}{Corollary}[theorem]
\newtheorem{lemma}[theorem]{Lemma}

\newcommand{\PKU}{Center on Frontiers of Computing Studies, Peking University, Beijing 100871, China}

\newcommand{\PKUCS}{School of Computer Science, Peking University, Beijing 100871, China}

\newcommand{\bytedance}{ByteDance Research, Fangheng Fashion Center, No.~27, North 3rd Ring West Road, Haidian District, Beijing 100098, China}

\newcommand{\HKU}{QICI Quantum Information and Computation Initiative, Department of Computer Science,
The University of Hong Kong, Pokfulam Road, Hong Kong}

\newcommand{\PKUMath}{Beijing International Center for Mathematical Research, Peking University, Beijing, China}

\begin{document}

% \linenumbers

% \preprint{APS/123-QED}
\title{Digital adiabatic evolution is universally accurate}

\author{Yangyu Lu}
\affiliation{\PKUCS}
\affiliation{\bytedance}
\affiliation{\PKU}

\author{Yifei Huang}
% \email{huangyifei.426@bytedance.com}
\affiliation{\bytedance}

\author{Dong An}
\email{dongan@pku.edu.cn}
\affiliation{\PKUMath}

\author{Qi Zhao}
\email{zhaoqi@cs.hku.hk}
\affiliation{\HKU}

\author{Dingshun Lv}
\email{lvdingshun@bytedance.com}
\affiliation{\bytedance}

\author{Xiao Yuan}
\email{xiaoyuan@pku.edu.cn}
\affiliation{\PKU}
\affiliation{\PKUCS}
% Force line breaks with \\

\begin{abstract}

Adiabatic evolution is a central paradigm in quantum physics. Digital simulations of adiabatic processes are generally viewed as costly, since algorithmic errors typically accumulate over the long evolution time, requiring exceptionally deep circuits to maintain accuracy. This work demonstrates that digital adiabatic evolution is intrinsically accurate and robust to simulation errors. We analyze two Hamiltonian simulation methods---Trotterization and generalized quantum signal processing---and prove that the simulation error does not increase with time.
% , for total evolution time $T$ and time step $\delta t$, the errors scale as $O(T^{-2} + \delta t^{2k})$ for $k$-th order Trotterization and $O(T^{-2})$ for GQSP, both decreasing with longer evolution times. 
We further show that accurate time-dependent adiabatic evolution can be achieved using only time-independent Hamiltonian-simulation algorithms.
% These results have direct implications for adiabatic quantum computation, quantum state preparation, and the solution of linear systems. 
Numerical simulations of molecular systems and linear equations confirm the theory, revealing that digital adiabatic evolution is substantially more efficient than previously assumed. 
Remarkably, our estimation for the first-order Trotterization error  can be $10^6$ times tighter than previous analyses for the transverse field Ising model even with less than 6 qubits.
The findings establish fundamental robustness of digital adiabatic evolution and provide a basis for accurate, efficient implementations on fault-tolerant—and potentially near-term—quantum platforms.

% Adiabatic evolution is a popular scheme for initial state preparation. However, the overhead of this process is not explored as thoroughly as other components of fault-tolerant algorithms like quantum phase estimation. In this work, we provide tight bounds for digital adiabatic evolution under two schemes: Trotterization and linear combination of unitaries~(LCU). The bounds have a novel form that connects the total evolution time $T$ to the gap properties at the beginning and end of the evolution. We show that the error under Trotterization does not accumulate with the total evolution time. Under the LCU scheme, the error cancels itself in the $T$-limit as $O(T^{-2})$, even with constant LCU orders. Under certain assumptions of the scheduling function, the infidelity can even scale as $O(e^{-T})$. We support our theoretical results with numerical findings. Our work highlights LCU as the preferred way to realize high-accuracy adiabatic quantum computation with low overhead. The results open up a doorway for potentially much more efficient state preparation in the fault-tolerant regime. \YH{This probably needs a rewrite.} %This echos the self-healing effective in analog adiabatic evolution [cite], where the difference is that there are no "healing" in the digital form due to Trotterization.
\end{abstract}

\date{\today}% It is always \today, today,


\maketitle

% \YL{change Title: like theory of trotter error}

% \YL{self-stablization]}

\begin{figure*}
\includegraphics[width=0.9\textwidth]{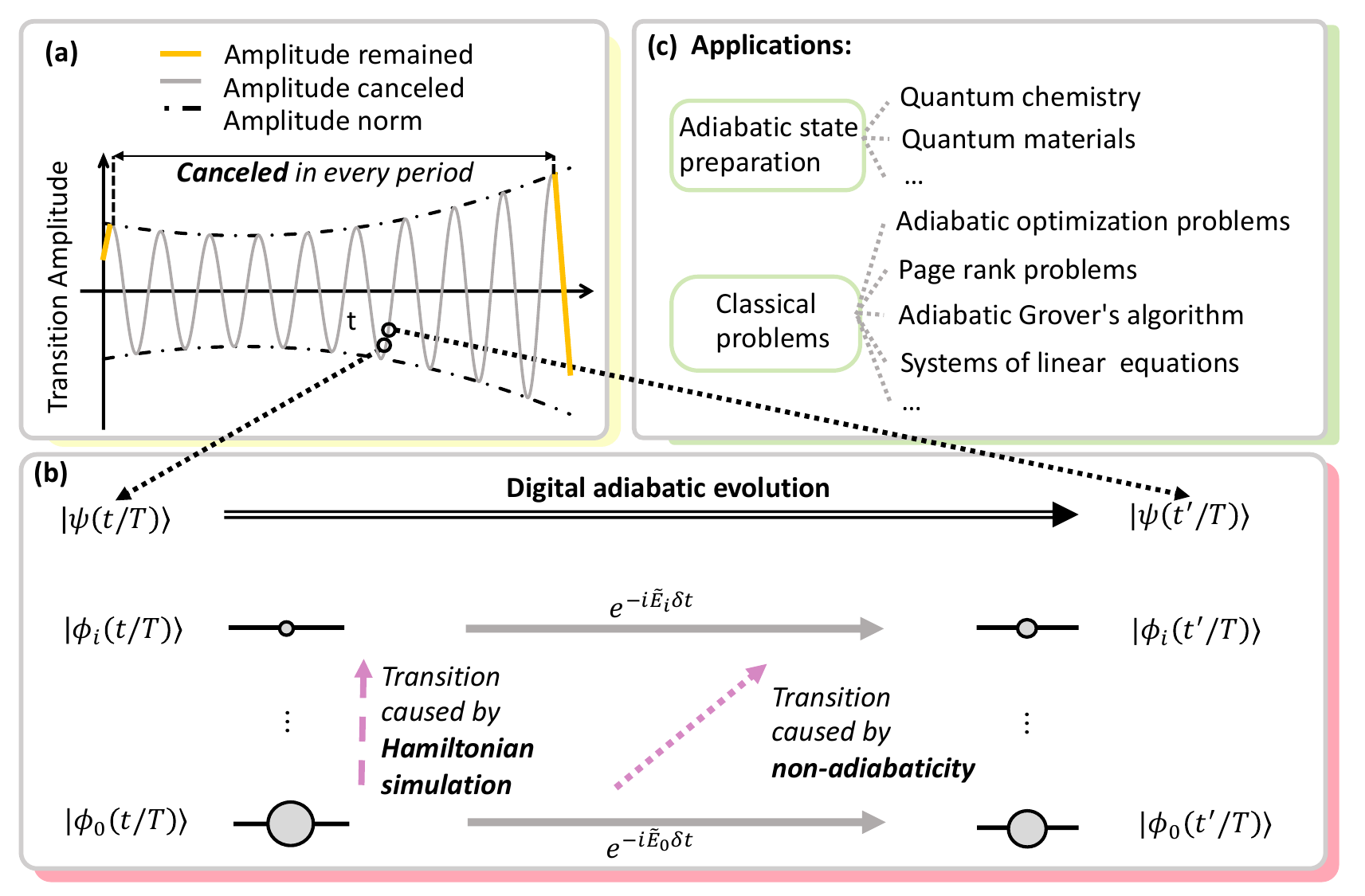}% Here is how to import EPS art
\caption{
Schematic of the error-cancellation mechanism and relevant applications. (\textbf{a}) Diagram of error cancellation in digital adiabatic evolution. The transition amplitude is bounded by a smooth envelope (black dash-dotted line) and oscillates frequently. As $T \to +\infty$, the integral of the transition amplitude cancels out within each period (grey line), leaving a residual part (orange line) independent of $T$. (\textbf{b}) Adiabatic time evolution from time $t$ to $t' = t + \delta t$. Errors from Hamiltonian simulation (pink dashed arrow) and non-adiabaticity (pink dotted arrow) excite the ground state $|\phi_0(t/T)\rangle$ to higher energy states $|\phi_i(t'/T)\rangle$. {Here $|\psi(t/T)\rangle$ is the state at time $t$, and $|\phi_i(t/T)\rangle$ is the $i$-th instantaneous eigenstate of the Hamiltonian $H(t/T)$ with eigenvalue $E_i(t/T)$. $\widetilde{E}_i(t/T)$ is the effective energy of the $i$-th eigenstate under the Hamiltonian simulation algorithm.} (\textbf{c}) Applications of our theory to adiabatic state preparation and classical problems.
}\label{fig:Fig1} 
\end{figure*}

\section{Introduction}
Adiabatic evolution has profound implications across various domains such as quantum state preparation~\cite{Abrams99, aspuru2005simulated, lee2023evaluating}, adiabatic quantum computation (AQC)~\cite{farhi2000quantum}, quantum annealing~\cite{kadowaki1998quantum,santoro2006optimization}, and quantum control~\cite{vandersypen2004nmr,zhou2017accelerated}. It relies on the principle of adiabaticity, where a quantum system remains in its ground state if the Hamiltonian changes slowly enough. For instance, in AQC, the ground state of a slowly varying Hamiltonian encodes the solution to a computational problem, offering a potential path to solving classically intractable problems \cite{van2001powerful,roland2002quantum,somma2012quantum,garnerone2012adiabatic}. Similarly, adiabatic quantum state preparation relies on gradually evolving a quantum system to prepare specific quantum states, which are crucial for tasks like quantum chemistry and materials science~\cite{RevModPhys.92.015003, cao2019quantum, Bauer_2020}. Remarkably, it has been proven that adiabatic evolution is equivalent to universal quantum computing~\cite{kempe2006complexity,RevModPhys.90.015002}, highlighting its critical role in the advancement of quantum technologies.

Digital quantum computers offer a powerful platform for simulating the adiabatic evolution of arbitrary quantum systems. By discretizing the evolution time, adiabatic processes can be decomposed into controllable sequences of single- and two-qubit gates~\cite{nielsen2010quantum}. However, this decomposition is inherently approximate, with algorithmic errors that decrease as circuit depth increases~\cite{childs2021theory}. 
Moreover, in long-time evolution processes, these algorithmic errors tend to accumulate over time~\cite{childs2018toward}. 
The adiabaticity condition, which requires long evolution times $T$ that scale at least polynomially with the inverse of the minimum energy gap in the system's Hamiltonian~\cite{jansen2007bounds,elgart2012note}, exacerbates this issue. Consequently, the resulting algorithmic error in digital adiabatic evolution also scales linearly with $T$, necessitating deeper quantum circuits to suppress the errors, thereby increasing computational costs. {Due to these challenges, which entail substantial computational overheads, digital simulations of adiabatic evolution are often regarded as impractical. As a result, they are typically treated as theoretical tools rather than viable algorithms for implementation on realistic quantum devices.
}
%Due to these challenges, digital simulations }of adiabatic evolution are often regarded as prohibitively expensive and are more commonly used as theoretical tools rather than practical algorithms on a realistic quantum computer. 
% \DA{``prohibitively expensive'' might be a bit exaggerated, because the overhead is only polynomial? shall we just say something like digital simulation is unfavorable due to these computational overheads?}

In this paper, we challenge this conventional view by showing the universal robustness of digital adiabatic evolution. Remarkably, based on a novel and systematic error analysis framework, we prove that the algorithmic errors in digital adiabatic simulations exhibit self-cancellation behavior, which does not necessarily scale with time. 
% This result extends previous findings on the self-healing of Trotterization errors with commuting Hamiltonian in adiabatic simulations~\cite{kovalsky2023self} but applies to a much broader range of practical scenarios and problems. 
% It also parallels the error interference phenomena observed in first-order Trotterization~\cite{tranDestructiveErrorInterference2020, laydenFirstOrderTrotterError2022, yiSpectralAnalysisProduct2021}, though with fundamentally different mechanisms and a much wider range of applicability.
We show that for two Hamiltonian simulation algorithms—Trotterization and GQSP~\cite{motlagh2024generalized}—the errors can be bounded as $O(T^{-2} + \delta t^{2k})$ for $k$th-order Trotterization (with time step $\delta t$) and $O(T^{-2})$, respectively, which both decease rather than increase with $T$. 
Furthermore, our theory reveals that simulating time-dependent adiabatic evolution can be effectively achieved using only simpler, well-understood time-independent Hamiltonian simulation algorithms, eliminating the dependence of auxiliary qubit overhead on the total evolution time or precision introduced by previous methods like truncated Dyson series~\cite{kieferova2019simulating} and qubitization~\cite{watkins2024time}.

% {Under certain scheduling functions, GQSP error bound can be even improved to $\bm{O(e^{-T})}$. Specifically, }

% \YL{}
% with time step $\delta t$ is the 
% , in contrast to the previously known self-healing bound of $O(T^{-2} + \delta t^2 T^{-2})$ for first order Trotterization~\cite{kovalsky2023self}.

Our findings have immediate applications in various quantum computing tasks, including adiabatic quantum state preparation and adiabatic-based linear system solvers~\cite{subacsi2019quantum}. The error analysis not only aids in understanding error behavior in digital adiabatic evolution but also in optimizing algorithmic parameters and performing more accurate resource assessments. We further validate our theoretical predictions through numerical simulations of molecular systems and linear equations. We demonstrate a $10^6$-fold improvement in the first-order Trotterization error bound when preparing the ground state of a transverse-field Ising model.
% , showing that digital adiabatic evolution can be achieved more efficiently than previously anticipated. 
This work paves the way for accurate and efficient digital adiabatic evolution, offering a promising route to harness quantum advantages with fault-tolerant quantum devices~\cite{preskill1998reliable}. By refining the computational overhead associated with adiabatic evolution, our results may even open new possibilities for the practical implementation of quantum algorithms in the era of noisy intermediate-scale quantum devices~\cite{preskill2018quantum}.

\section{Main results}

% In this section, we present our main results.
We begin by reviewing the overall framework of digital adiabatic evolution, focusing on the errors arising from non-adiabatic transitions and imperfect Hamiltonian simulation. Next, we introduce a general method that bounds the infidelity of a digital adiabatic process for general Hamiltonian simulation algorithms. Finally, we apply our results to two specific Hamiltonian simulation methods—Trotterization and GQSP—providing concrete bounds for each.
% We only summarize the main results in the main text and refer to the Appendix for proofs and details.     
%% Recover later

% \subsection{Adiabatic quantum computation with Hamiltonian simulation}
\bigskip
\noindent\textbf{\emph{Background}}\\
We focus on the linear adiabatic path from an initial Hamiltonian $H_i$ to a final Hamiltonian $H_f$
\begin{equation}
H({t}/{T})=[1-u({t}/{T})]H_i + u({t}/{T})H_f,
\end{equation}
where $u(x)$ is a scheduling function satisfying $\dot{u}(x) \geq 0$, $u(0) = 0$, and $u(1) = 1$. 
Note that the analysis applies to arbitrary adiabatic evolution. 
The evolution time $T$ must be sufficiently large, specifically $T \gg g_{\text{min}}^{-2}$ in the general case, to ensure that the process remains adiabatic. Here, $g_{\text{min}}$ denotes the minimum eigenvalue gap between the ground state and the first excited state of $H(t/T)$. In AQC, the initial Hamiltonian $H_i$ is typically chosen to have a simple, analytically solvable ground state, while the problem to be solved is encoded in the final Hamiltonian $H_f$, allowing the system to adiabatically evolve from the simple initial ground state to the complex final ground state.

We denote the state at time $t$ as $
|\psi(t/T)\rangle = \mathcal{T}e^{-i\int_0^tH(t'/T)dt'}|\psi(0)\rangle$, 
% \DA{$\mathcal{T}e^{-i\int_0^tH(s/T)ds}|\psi(0)\rangle$}
where $\mathcal{T}$ is the time-ordering operator. The target solution is then given by $|\psi(1)\rangle$ at $t = T$.
To implement the adiabatic evolution on a quantum computer, we can approximate the evolution using a sequence of unitary operators as follows~\cite{van2001powerful}:
\begin{equation}\label{eq:slice}
\begin{aligned}
|\psi(1)\rangle &= \prod_{m=1}^{r} U\left(\frac{m}{r}\right) |\psi(0)\rangle,
\end{aligned}
\end{equation}
where $U(m/r) := \exp[-iH(m/r)\delta t]$, $\delta t$ is the time step size, and $r = T/\delta t$ is the number of time steps. Here, we have used a time-independent approximation of the time-dependent Hamiltonian evolution. However, we prove below that this approximation yields negligible errors.

Consider a general Hamiltonian simulation method, which approximates the time evolution operator $U(m/r)$ at step $m$ with a quantum circuit described by operator $\widetilde{U}(m/r)$ as:
\begin{equation}
\widetilde{U}(m/r) = U(m/r) + U_{\text{res}}(m/r)\delta t^{k+1} + O(\delta t^{k+2}),
\end{equation}
where $U_{\text{res}}$ is a residual term and $k\ge 1$ is a certain integer. It is important to note that the circuit operator $\widetilde{U}$ may not be unitary in general, such as in the case of GQSP. For this reason, we introduce a normalization factor at time step $m$ defined as 
% \DA{$\phi_0$ (and $\phi_i$ used later) is not defined before this equation}
\begin{equation}\label{Eq:amr}
a\left(\frac{m}{r}\right) = \sqrt{\langle \phi_0\left(\frac{m}{r}\right)|\widetilde{U}^{\dagger}\left(\frac{m}{r}\right)\widetilde{U}\left(\frac{m}{r}\right)|\phi_0\left(\frac{m}{r}\right)\rangle}.
\end{equation}
Here $|\phi_i(x)\rangle$ is the instantaneous eigenstate of the time-dependent Hamiltonian $H(x)$ with the $i$-th lowest eigenvalue $E_i(x)$. 

Considering the quantum circuit operator $\widetilde{U}$, if it is unitary, we can define its effective Hamiltonian $\widetilde{H}$ and the corresponding effective energy $\widetilde{E}_i$. Otherwise, we can still generalize this to the non-unitary case,
%\DA{I'm a bit confused: is $\widetilde{U}$ the same numerical operator which could be non-unitary? The following definition seems to require $\widetilde{U}$ to be unitary. }
by defining the effective energy at step $m$ as $\widetilde{E}_i(m/r):= -\arg(\langle \phi_i|\widetilde{U}|\phi_i\rangle)/\delta t$ and the average effective energy as $\bar{E}_i(m/r) = m^{-1}\sum_{m'=1}^m \widetilde{E}_i(m'/r)$. It is worth noting that in our case, $\delta t$ is taken to be sufficiently small, which allows us to avoid the winding number issue. The effective gap and average effective gap are then defined as $\widetilde{\Delta}_i(m/r) = \widetilde{E}_i(m/r) - \widetilde{E}_0(m/r)$ and $\bar{\Delta}_i(m/r) = m^{-1}\sum_{m'=1}^m \widetilde{\Delta}_i(m'/r)$, respectively. 
% \DA{Here might be a winding number issue if $\delta t$ is large. Perhaps also mention here that $\delta t$ will be taken sufficiently small? } 
We note that since the time-dependent Hamiltonian $H(x)$ is a function defined at the interval $x\in [0,1]$, the variables  $U(x)$, $\widetilde{U}(x)$ and $\widetilde{\Delta}_i(x)$ could be accordingly defined continuously at the interval $x\in [0,1]$ with the limit of $T\rightarrow \infty$. Sequentially, we define a continuous version of $\bar{\Delta}_i(x)$ as $\bar{\Delta}_i(x) = x^{-1}\int_0^x \widetilde{\Delta}_i(x')dx'$.

% \vspace{0.4cm}
\bigskip

\noindent\textbf{\emph{General results}}\\
Now, we are ready to show our general result for characterizing the errors in digital adiabatic evolution. We focus on the infidelity between the digitally evolved state $|\psi(m/r)\rangle$ and the exact ground state $|{\phi_0(m/r)}\rangle$ of the time-dependent Hamiltonian $H(m/r)$, defined as
\begin{equation}
\begin{aligned}
\mathcal{I}\left(\frac{m}{r}\right) := 1 - \left|\langle\psi\left(\frac{m}{r}\right)|\phi_0\left(\frac{m}{r}\right)\rangle\right|^2 \\
= \sum_{i \neq 0} \left|\langle\psi\left(\frac{m}{r}\right)|\phi_i\left(\frac{m}{r}\right)\rangle\right|^2.
\end{aligned}
\end{equation}
While we focus on the ground state in this work, the analysis can be extended to general excited states.

In digital adiabatic evolution at step $m$, the infidelity arises from non-adiabatic transitions, algorithmic errors in Hamiltonian simulation, and amplitude rescaling, as illustrated in Fig.~\ref{fig:Fig1}(b). Specifically, the transition amplitude $A_i(m/r)$ at step $m$ is defined as:
\begin{equation}
\begin{aligned}
    A_i(x)&:=\langle e^{-i\bar{E}_i m\delta t}\phi_i(x+1/{r})|\widetilde{U}(x)|e^{-i\bar{E}_0  m\delta t}\phi_0(x)\rangle \\
    &=R_i(x)\exp(i\bar{\Delta}_i(x)m\delta t).
\end{aligned}
\end{equation}
Here $x=m/r$, $\bar{\Delta}_i(m/r)m\delta t$ represents the phase difference between the two eigenstates, and $R_i(m/r)$ is the phaseless transition amplitude, given by:
\begin{equation}
\begin{aligned}
R_i(m/r):&=\langle \phi_i((m+1)/r))|\phi_0(m/r)\rangle \\&
+ \langle \phi_i(m/r)| U_{\text{res}}|\phi_0(m/r)\rangle\delta t^{k+1},
\end{aligned}
\end{equation}
where the first term corresponds to the non-adiabatic transition, and the second term accounts for the Hamiltonian simulation algorithmic error.
In addition to eigenstate transitions, the amplitude on the excited state $|\phi_i(m/r)\rangle$ is rescaled by the factor $|\widetilde{U}_{ii}(m/r)|/a(m/r) = \exp(-\Lambda_i(m/r)\delta t)$, where $\widetilde{U}_{ii}(m/r) := \langle \phi_i(m/r)| \widetilde{U}(m/r)|\phi_i(m/r)\rangle$. This rescaling is due to the non-unitarity of $\widetilde{U}$. The decay rate $\Lambda_i(x)$ can be calculated by:
\begin{equation}
\begin{aligned}
\Lambda_i(x) := \frac{1}{\delta t} \ln\left(\frac{\sqrt{|\langle \phi_0(x)|\widetilde{U}^{\dagger}(x)\widetilde{U}(x)|\phi_0(x)\rangle|}}{|\langle \phi_i(x)|\widetilde{U}(x)|\phi_i(x)\rangle|}\right).
\end{aligned}
\end{equation}
We also define the average decay rate $\bar{\Lambda}_i(m/r) = m^{-1}\sum_{m'=1}^m \Lambda_i(m'/r)$ and its continuous version $\bar{\Lambda}_i(x)=x^{-1}\int_0^x \Lambda_i(x')dx'$.

The infidelity of the state at time $t = T$ can be expressed as
\begin{equation} 
\begin{aligned}
\mathcal{I}\left(1\right) &\approx
\sum_{i \neq 0} \left|\sum_{m =1}^r R_i(m/r)e^{i\bar{\Delta}_i(m/r)m\delta t} \right|^2,
\end{aligned}
\end{equation}
where we have neglected higher-order errors of $O(\delta t^{k+2})$ in each summation. Intuitively, the frequency of $R_i(m/r)$ is proportional to $T^{-1}$, while the frequency of $\exp(i\bar{\Delta}_i(m/r)m\delta t)$ is independent of $T$. In the large $T$ limit, the transition amplitude $A_i(m/r)$ is modulated by a smooth envelope but oscillates rapidly, as illustrated in Fig.~\ref{fig:Fig1}(a). The integrals over each oscillation period tend to cancel out, leaving the sums of the amplitudes near the ends as the dominant terms. Such cancellations can be rigorously proven with the help of oscillatory integrals~\cite{degani2006rcms,iserles2005efficient} and are the origin of the robustness of digital adiabatic evolution. We encapsulate it in the following theorem.

\smallskip
\noindent \textbf{{Theorem 1 (General error cancellation in digital adiabatic evolution)}.} %In an AQC process, at step $m$, we use $\widetilde{U}(m/r)$ to approximate the exact time evolution operator $U(m/r)$, where $\widetilde{U}(x)=U(x)+U_{\text{res}}(x)\delta t^{k+1} + O(\delta t^{k+2})$. 
\emph{Considering small $\delta t$ and large $T$, if the following condition
\begin{equation}\label{eq:condition}
    \left\{
    \begin{aligned}
        &\exp{(-\overline{\Lambda}_i(x)T)}=O(1),~~\forall x\in [0,1]\\
         &\widetilde{\Delta}_i(x)>0, ~~\forall x\in [0,1],
    \end{aligned}
    \right.
\end{equation}
is satisfied, then the infidelity $\mathcal{I}(1)$ in a digital adiabatic evolution scales as:
\begin{equation}\label{eq:bound1}
    \mathcal{I}(1)=O\left(\beta_{\text{ad}}^2T^{-2}+ \beta_{\text{sim}}^2\delta t^{2k}\right),
\end{equation}
where $\beta_{\text{ad}}$ and $\beta_{\text{sim}}$ are coefficients for non-adiabatic error and Hamiltonian simulation error, which are independent of $T$ and $\delta t$.}
 
% \textcolor{green}{[revised]}
To interpret the conditions in Eq.~(\ref{eq:condition}), the first condition ensures that the average decay rate $\bar{\Lambda}_i(x)$ is non-negative or its magnitude is sufficiently small compared to $T^{-1}$, which prevents the amplitude on the excited state from increasing uncontrollably. The second condition requires a non-zero effective gap $\widetilde{\Delta}_i(x)$. This condition is satisfied by many Hamiltonian simulation algorithms. For example, in Trotterization, it is easy to validate $\Lambda_i(x) \geq 0$ and $\widetilde{\Delta}_i(x) > 0$ for sufficiently small $\delta t$. We also prove in the Appendix \ref{sec:condition_of_GQSP} that these conditions hold for GQSP.
The theorem serves as the basis for analyzing the universal robustness of digital adiabatic evolution with arbitrary Hamiltonian simulation techniques. In the following, we consider two examples of Trotterization and GQSP, which respectively demonstrate weak and strong error cancellation, and discuss surprising exponential error cancellation with optimized adiabatic path. 

\smallskip

% \noindent \Del{The two conditions correspond to the sub-normalization of eigenstates and a nonzero gap in the effective Hamiltonian, which are satisfied by Trotterization and certain orders of LCU, as discussed in detail below. For other Hamiltonian simulation methods, like quantum signal processing (QSP)~\cite{low2017optimal}, it remains unclear how to bound $\Lambda_i$, making it an interesting topic for future research to apply the application of Theorem 1 in QSP. Our analysis is based on a time-independent Hamiltonian simulation algorithm. However, it’s important to note that even when using time-dependent Hamiltonian simulation methods to approximate the time-dependent evolution operator $U(m/r) = \mathcal{T}\exp(-i\int_{m\delta t}^{(m+1)\delta t} H(t’)dt’)$, Theorem 1 remains applicable. This insight leads to an intriguing conclusion: for digital adiabatic evolution, time-independent Hamiltonian simulation methods are sufficient to achieve the same scaling in terms of $T$ and $\delta t$ as time-dependent methods. }

\bigskip

\noindent\textbf{\emph{Trotterization}}\\
Now, let’s consider Trotterization~\cite{suzukiGeneralTheoryFractal1991} as the Hamiltonian simulation method and derive the weak error cancellation in digital adiabatic evolution. Here, ``weak'' refers to the fact that while the infidelity does not accumulate with $T$, the Hamiltonian simulation error still depends on the time step $\delta t$. In contrast, a stronger version of error cancellation, such as with the GQSP method, would be independent of $\delta t$.
% The time-evolution operator $U(m/r)$ in Eq.~(\ref{eq:time_evolution}) can be approximated by the Trotter formula. 

We break down the Trotterization process into two stages: \textit{primary Trotterization} and \textit{sub-Trotterization}. In primary Trotterization, the unitary operator $U(m/r)=\exp(-i\{[1-u(m/r)]H_i+u(m/r)H_f\}\delta t)$ is approximately decomposed into two components $\exp(-i[1-u(m/r)]H_i\delta t)$ and $\exp(-iu(m/r)H_f\delta t)$ as follows:

\begin{equation}
\begin{aligned}
U(m/r)&\approx \widetilde U(m/r)\\
&= e^{-i[1-u(m/r)]H_i\delta t}e^{-iu(m/r)H_f\delta t}.
\end{aligned}
\end{equation}
Each part is then assumed to evolve jointly. This decomposition is particularly effective when  $H_f$ is classical, where its evolution incurs no sub-Trotterization error, or in analog quantum simulations where the joint evolution of $H_f$ is naturally feasible~\cite{blatt2012quantum,gross2017quantum,aspuru2012photonic,daley2022practical}. Assuming primary Trotterization, the ``self-healing'' phenomenon has been discovered, where the infidelity scales as  $O(\delta t^2T^{-2}+T^{-2})$~\cite{kovalsky2023self}. Compared to our weak error cancellation result, the Trotterization error part is also suppressed by $T^{-2}$, which becomes negligible as long as the ratio $\delta t/T$ is small. We will show that such a stronger error cancellation would vanish when sub-Trotterization of $H_f$ is considered.

Specifically, for general $H_f$, we need to consider its Trotterization decomposition, referred to as sub-Trotterization. Without loss of generality, it is assumed that $H_i$ and $H_f$ can be decomposed into the form as $H=\sum_{j}c_jP_j$, where $P_j$s are $n$-qubit Pauli operators. In this case, we can approximate the joint evolution $e^{-iH\delta t}$ by using $\prod_j\exp(-ic_jP_j\delta t)$, which allows for efficient realization of each term as a sequence of single-Pauli rotations on a digital quantum computer. We note that higher-order Trotter-Suzuki formulas~\cite{suzukiGeneralTheoryFractal1991} can also be applied here.
We show that when both primary Trotterization and sub-Trotterization are considered, the overall infidelity no longer decreases with  $T$, but crucially, it does not accumulate with  $T$  either. We summarize this result as follows.

% when considering both primary Trotterization and sub-Trotterization, the overall infidelity will not decrease with $T$ any more. However, it does not accumulate with $T$ as well, which still outperforms the general Trotter error bound $O(T^2\delta t^{2k})$ for a unitary time evolution~\cite{layden2022first}. We summarize the result as follows. 

% and is thoroughly discussed by the self-healing PRL paper, we will pay more attention to AQC with sub-Trotterization.

% In general, the Pauli terms in $H_i$ and $H_f$ don't commute with each other. %Hence, one expect a sub-Trotterization error at $t=0$ and $t=T$ like that of primary Trotterization. 
% In this case, the self-healing effect breaks down. Instead, our result shows that although Trotter error will not decrease with $T$ anymore, it does not accumulate with $T$ as well, which still outperforms the general Trotter error bound $O(T^2\delta t^{2k})$ for a unitary time evolution~\cite{layden2022first}. 

\smallskip

\noindent
\textbf{{Theorem 2 (Weak error cancellation with Trotterization).}}
\emph{For large $T$ and small $\delta t$, the infidelity $\mathcal I(1)$ in digital adiabatic evolution with the $k$-th order Trotterization formula scales as $O(\beta_{\text{ad}}^2T^{-2}+\beta_{\text{tro}}^2\delta t^{2k})$, where $\beta_{\text{ad}}$ and $\beta_{\text{tro}}$ are coefficients for the non-adiabatic error and Trotter error, which are irrelevant to $T$ and $\delta t$.}

\smallskip

\noindent Theorem 2 could be directly obtained from the more general result Theorem 1. Specifically, since $\widetilde{U}$ is unitary in this case, the condition $\exp{(-\overline{\Lambda}_i(x)T)}=O(1)$ is satisfied naturally as $\Lambda_i(x)\geq0$. To ensure $\widetilde{\Delta}_i(x)>0$, we need $\delta t^{k+1}$ to be sufficiently small compared to $g_{\text{min}}$. Otherwise, one may encounter the ``gap closure'' phenomenon~\cite{yi2021success}, which will lead to the break down of the adiabaticity condition.

% Our result shows that although trotter error will not decrease with $T$ any more, it does not accumulate with $T$ as well, which is different from the general Trotter error bound $O(T^2\delta t^{2k})$ for a unitary time evolution~\cite{}. %This is ensured by some features of AQC process, and whether the result can be extended to other time evolution process is under further investigation. We answer the question what makes AQC different in Appendix \ref{sec:what_makes_AQC_different}.

While only considering primary Trotterization, our result reduces to the self-healing case~\cite{kovalsky2023self}. We observe that for a great number of classical problems, such as the glued trees problem~\cite{somma2012quantum} and adiabatic optimization problems~\cite{van2001powerful}, the primary Trotterization is sufficient, and the theory of ``self-healing'' applies.
% \DA{need to revise the wording in this paragraph as the Corollary has been removed} 
However, for more complex quantum many-body problems with intricate interactions, sub-Trotterization must be considered, leading us to Theorem 2. Interestingly, the nature of error cancellation differs between classical and quantum problems. In classical problems, where Hamiltonians have commuting terms and are easier to handle, a stronger self-healing phenomenon occurs. In contrast, for more challenging quantum problems, only weaker error cancellation is observed. Nevertheless, this still outperforms the conventional Trotter error bound  $O(T^2\delta t^{2k})$~\cite{layden2022first}, which generally increases with the total simulation time $T$. 

\smallskip

Theorem 2 not only reveals the fundamental phenomenon of error robustness in digital adiabatic evolution, but also can be applied to optimize the algorithm parameters $\delta t, T$ given a restricted circuit depth $d$. We summarize the result as follows. 

% Furthermore, Corollary 2.1 gives a strategy to minimize the upper bound of infidelity with fixed circuit depth $d$ by adjusting the parameters $\delta t, T$. 

\smallskip
\noindent\textbf{{Corollary 2.1 (Optimization of digital adiabatic evolution under fixed circuit depth, informal).}}
\emph{
For a large fixed circuit depth $d$, we can pick a set of optimal parameters $\delta t, T$, such that the infidelity of the state prepared by the digital adiabatic process with the $k$-th order sub-Trotterization is minimized. The optimal parameters scale as
% While $k=1$ or $k\geq 2$ and $k$ is an even number:
\begin{equation}
\begin{aligned}
&T_{opt}=O(d^{\frac{k}{k+1}})~{\rm and}~
\delta t_{opt}=O(d^{-\frac{1}{k+1}}),
\end{aligned}
\end{equation}
and the optimal infidelity scales as
\begin{equation}
    \mathcal I(1) = O\left(d^{-\frac{2k}{k+1}}\right),
\end{equation}
where $k=1$ or $k\geq 2$ and $k$ is an even number.
}
\smallskip

% \noindent The exact coefficients are given in Appendix \ref{sec:proof_of_corollary2.1}. 

\noindent In conventional analysis, we generally need to choose a sufficiently small $\delta t$ and increase $d$ and $T$ simultaneously to control the errors~\cite{sugisaki2022adiabatic}. What we find here indicates that there is an optimal choice of the time step $\delta t$ and the total evolution time $T$ for adiabatic evolution with fixed circuit depth $d$. The optimal parameters are also different when only considering the error from  primary Trotterization. However, it is optimal only for problems where the initial Hamiltonian and the final Hamiltonian have commuting terms. 

% The optimal parameters are also different when considering primary Trotterization, however, it is optimal only for classical problems.

% \YH{TAKE THIS OUT FIRST.}% We remark that Corollary 2.1 also shows that if the circuit depth is small, low-order Trotterization is a good choice for high accuracy. While the circuit depth is large, the order of Trotterization should be increased.  On the theoretical level, the previous work only attains the scaling of the infidelity upper bound, but not coefficients because the authors expand the error term into trigonometric series and only analyse the contribution of one certain term. The proof also utilizes the solution for time-independent Hamiltonian with time-dependent perturbation, but in fact, the Hamiltonian is time-dependent. Only when the energy gap is invariant through the evolution, the solution can be used directly to AQC process. However, the energy gap is not invariant in general, so their proof is not rigorous. Our work gives a rigorous derivation of the infidelity upper bound for AQC with only primary Trotterization  (the self-healing case) and AQC with sub-Trotterization.

\bigskip

\noindent 
\textbf{\emph{GQSP}}\\
Next, we consider a more advanced Hamiltonian simulation method based on the GQSP. Unlike Trotterization, the GQSP approach eliminates eigenstate transitions due to algorithmic errors. As a result, it only experiences non-adiabatic errors. Consequently, as we demonstrate below, digital adiabatic evolution using the GQSP method exhibits strong error cancellation, with the infidelity monotonically decreasing as  $T^{-2}$.

%In this section, we will show that the error of adiabatic quantum computation with LCU does self-heal in the large $T$ limit.

Assume that the Hamiltonian can be decomposed as $H(x)=\sum_1^L \alpha_l(x)U_l$ for Hermitian unitaries $U_l$ and real-valued $\alpha_l(t)$. To use the GQSP approach with a $K$-th order truncated polynomial for Hamiltonian simulation, 
%\DA{$K$-th order GQSP sounds a bit weird terminology, as $K$ will be chosen adaptively according to other parameter. Perhaps ``To use the GQSP approach with a $K$-th order truncated polynomial for Hamiltonian simulation''?}
we can choose $\widetilde{U}$ as: 

\begin{equation}
    \widetilde{U}(m/r)=P_K(e^{i\arccos(H(m/r)/\alpha)}).
\end{equation}
Here we define:
\begin{equation}
    P_K(x):=\sum_{k=-K}^K (-i)^kJ_k(\alpha \delta t)x^k,
\end{equation}
and $J_k$ is the $k$-th order Bessel function. Here $\alpha=\max_{x\in[0,1]}\sum_l^L|\alpha_l(x)|$ is a factor introduced to make sure $|E_i(m/r)/\alpha|\leq 1$. Notice that $\widetilde{U}$ here is not unitary, so we have to introduce auxiliary qubits and carry out post selections. Details of GQSP are given in Appendix \ref{sec:details_of_GQSP}.

To satisfy the condition Eq.~(\ref{eq:condition}) in Theorem 1 and keep the possibility of failure within $\delta$, we need to ensure that the order $K$ scales as $\widetilde{O}(\alpha \delta t + \log(r/\delta))$. Moreover, since $\widetilde{U}$ here is functions of $H$, it is diagonal under the basis $\{\phi_i(m/r)\}$ and the algorithmic error of GQSP will not give rise to transition between each eigenstate. For this reason, the infidelity will scale as $O(T^{-2})$, just like the case where we assume Hamiltonian simulation can be carried out with no error. The analysis above can be summarized by the following theorem:

\smallskip
\noindent
\textbf{{Theorem 3 (Strong error cancellation with GQSP).}}
\emph{
For large $T$, small $\delta t$, and $K\sim O(\log(r/\delta)/\log\log(r/\delta))$, GQSP can simulate digital adiabatic evolution using $O(\log (L))$ anxiliary qubits while the infidelity $\mathcal I(1)$ scales as $O(\beta_{\text{ad}}^2T^{-2})$ and the probability of failure will be less than $\delta$. Here $\beta_{\text{ad}}$ is independent of $T$ and $\delta t$.}

\smallskip

% \noindent \Del{In conventional GQSP, the order is typically chosen as $K = O(\alpha t+\log(1/\epsilon)/\log\log(1/\epsilon))$ for time-independent Hamiltonian evolution~\cite{}. If one wants to use GQSP for a general time-dependent Hamiltonian simulation, after dividing the time $t$ into $r$ slices of $\delta t$, to keep the infidelity within $\epsilon$,  the order should be chosen as $K = O(\alpha \delta t+\log(r/\epsilon)/\log\log(r/\epsilon))$. Contrarily, in adiabatic evolution, Theorem 3 indicates  that the order $K$ the order $K$ should be $O(\alpha \delta t + \log(r/\delta)/\log\log(r/\delta))$ with the probability of failure under $\delta$. While $\delta >> \epsilon$, the circuit depth is smaller than that in a general time-dependent Hamiltonian simulation. The average query times $O(r/(1-\delta)[\alpha \delta t + \log(r/\delta)/\log\log(r/\delta)])$ is also smaller than that in the general cases $O(r/(1-\epsilon)[\alpha \delta t + \log(r/\epsilon)/\log\log(r/\epsilon)])$. For small average query times, we recommend $0.1<\delta<0.5$ if $\epsilon=10^{-3}$.}

We would like to emphasize that the result of strong error cancellation is very intriguing. First, GQSP is originally designed for time-independent Hamiltonians, and {it works here because we discretize the time and apply GQSP to simulate the  Hamiltonian corresponding to the end point of each time step}. Intuitively, this would lead to an approximation error of $\delta t^2$ as in the first-order Trotterization because we neglect the time dependence of the Hamiltonian. This explains why we generally need much more sophisticated methods such as truncated Dyson series~\cite{kieferova2019simulating} to handle time-dependent Hamiltonians. However, our result indicates that time-independent GQSP is sufficient for accurate simulation of time-dependent adiabatic evolution, making its simulation much simpler in theory and easier in practice. Second, the truncation order $K$ in conventional GQSP generally increases with the simulation accuracy $\epsilon$ and total simulation time $T$. Our result further shows that the strong error cancellation effect allows us to {consider $K$ being independent of $\epsilon$.} These surprising effects thus broadens our understanding of time-dependent Hamiltonian simulation and greatly simplifies the quantum algorithm design and cost for practical adiabatic evolution.

%The strong error cancellation in GQSP makes it possible to use time-independent Hamiltonian simulation methods to achieve the infidelity $\epsilon$ with query complexity logarithmic in $1/\epsilon$. In comparison with previous methods, like truncated Dyson series~\cite{kieferova2019simulating}, the number of auxiliary qubits required will not increase with $\epsilon$ and $T$. 

% \textcolor{green}{[revised]}

\begin{table*}[htbp]
\centering
\caption{A summary of our results and a comparison with the leading time-dependent Hamiltonian simulation methods for adiabatic evolution. Here $\epsilon$ and $\delta$ are upper bounds for infidelity and the possibility of faliure respectively. $k$ and $K$ are the order of Trotterization and truncated polynomial used in GQSP respectively. We assume that $H = \sum_{l=1}^L \alpha_l(t) U_l$ for Hermitian unitaries $U_l$ and real-valued $\alpha_l(t)$. {``STDA'' is short for ``Specialized Time-Dependent Algorithms'', which means the algorithm is tailored for time-dependent Hamiltonian simulation to eliminate the error from time-dependence.} }
\label{tab:quantum_sim_summary}
% Changed from {llll} to {lllll} to add the new column
\begin{tabular}{lllll} 
\toprule
% Added the 'TDS' column header
\textbf{Method} & \textbf{Error Bound} & \textbf{Query Complexity} & \textbf{Auxiliary Qubits} & \textbf{STDA} \\
\midrule
% Added the new 'TDS' data to each row
Trotter \cite{wiebe2010higher} & $O(T^{-2}+T^2\delta t^{2k})$ & $O(T/\epsilon^{1/k})$ & 0 & Yes \\
\addlinespace
QDrift \cite{berry2020time} & $O(T^{-2}+T^2\delta t^{2})$ & $O(T/\epsilon)$ & 0 & Yes \\
\addlinespace
Trotter (This work) & $O(T^{-2}+\delta t^{2k})$ & $O(T/\epsilon^{1/(2k)})$ & 0 & No \\
\addlinespace
\hline
\addlinespace
Dyson \cite{kieferova2019simulating} & $O(T^{-2}+T^2\delta t^{2K}/((K+1)!)^2)$ & $\widetilde{O}(T\log(\max{(r/\delta,r/\epsilon)}))$ & $\widetilde{O}(\log(T\delta t/\epsilon) + \log(L))$ & Yes \\
\addlinespace
Qubitization \cite{watkins2024time} & $O(T^{-2}+T^2\delta t^{2K}/((K+1)!)^2)$ & $\widetilde{O}(T\log(\max{(r/\delta,r/\epsilon)}))$ & $\widetilde{O}(\log(T\delta t/\epsilon)+\log(L))$ & Yes \\
\addlinespace
GQSP (This work) & $O(T^{-2}) $ & $\widetilde{O}(T\log(r/\delta))$ & $O(\log(L))$ & No \\
\bottomrule
\end{tabular}
\end{table*}

Using Theorem 3, we can further derive the query complexity and auxiliary qubit requirements for simulating adiabatic evolution with GQSP to achieve a desired infidelity $\epsilon$ and failure probability $\delta$:

\smallskip
\noindent\textbf{{Corollary 3.1}}
\emph{
For a digitally simulated adiabatic evolution with large total time $T$ and small step size $\delta t$, the query complexity of our GQSP method scales as $\widetilde{O}(T\log(r/\delta))$ to achieve an infidelity under $\epsilon$ and failure probability within $\delta$ using $O(\log(L))$ auxiliary qubits.
}
\smallskip

\noindent Notice that the number of auxiliary qubits required is reduced from $O(\log(T\delta t/\epsilon)+\log(L))$ in the truncated Dyson series~\cite{kieferova2019simulating} and qubitization~\cite{watkins2024time} to $O(\log(L))$ in our GQSP method as we stop encoding Hamiltonian at different times together as in the truncated Dyson series and qubitization, which requires an additional $\log(T\delta t/\epsilon)$ auxiliary qubits to store the time information.\\

In Table~\ref{tab:quantum_sim_summary}, we present a summary of our findings and compare them to leading quantum simulation methods for time-dependent Hamiltonians: Trotterization, QDrift, the truncated Dyson series (Dyson) and Qubitization. Among the methods without auxiliary qubits, Trotterization in our work achieves the best query complexity. For methods requiring auxiliary qubits, our GQSP method outperforms Dyson and Qubitization.

\bigskip

\begin{figure*}
\includegraphics[width=1\textwidth]{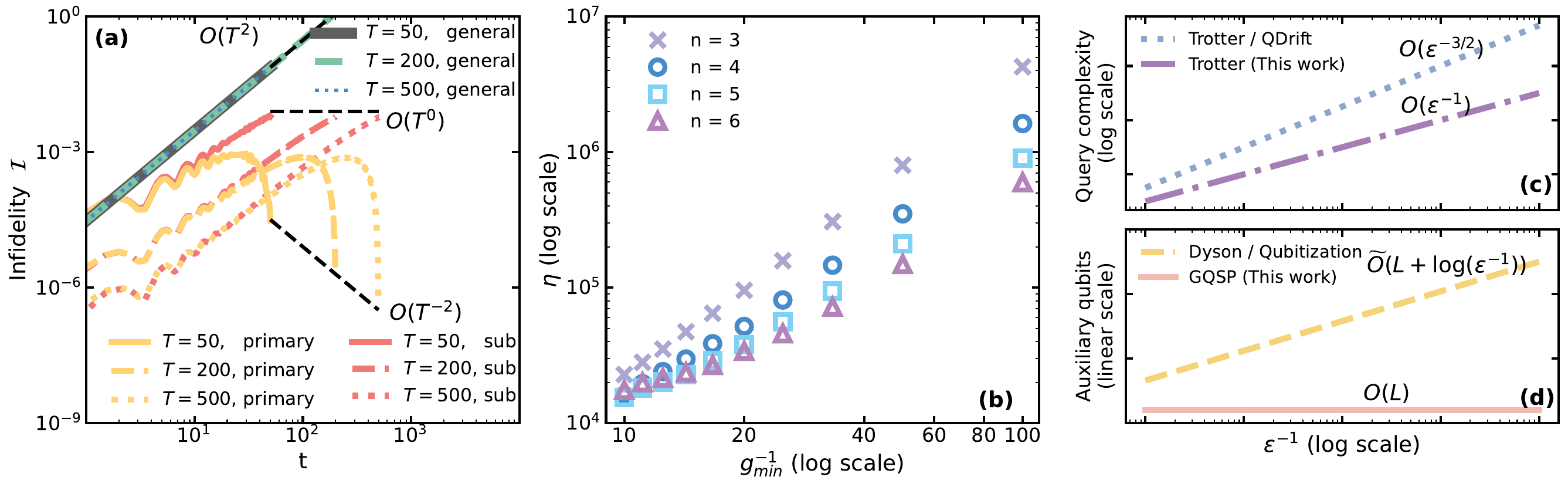}% Here is how to import EPS art
\caption{
Comparison of different infidelities, query complexity, and auxiliary qubits scaling. Here, first-order Trotterization is studied. (\textbf{a}) Infidelity of digital adiabatic evolution with primary Trotterization (yellow), sub-Trotterization (red), and  Trotter error for a general time-dependent Hamiltonian evolution, with $T=50$ (solid line), $T=200$ (dashed line), and $T=500$ (dotted line). {(\textbf{b}) The improvement factor $\eta$, defined as the ratio of the former to our proposed Trotterization error bound, as a function of $g_{\text{min}}^{-1}$ for systems with $n=3$ to $6$ qubits. Here we set $H_i = -\sum_{i=1}^n X_i$ and $H_f = \sum_{i=1}^{n-1} (-2Z_iZ_{i+1} + g\cdot Z_n+X_i)$. By varying $g$ from $0.01$ to $0.1$ we can change the minimum gap $g_{\text{min}}$ within the range [0.02, 0.31]. $T$ and $\delta t$ are chosen as the critical values required to ensure the adiabatic infidelity $\mathcal{I}_{\text{ad}}$ and simulation infidelity $\mathcal{I}_{\text{sim}}$ (seriously defined in Eq.~(\ref{eq:infi})) are both below the threhold of $0.001$. Here, we consider first-order Trotterization.}(\textbf{c}) Comparison of query complexity scaling with respect to $\epsilon$ between previous works and our work. Here we illustrate the performance of Trotterization (previous work), QDrift and Trotterization (this work). (\textbf{d}) Comparison of auxiliary qubits scaling with respect to $\epsilon$ between previous works and our work. Here we illustrate the performance of truncated Dyson series, Qubitization and GQSP (This work). Here we use the relationship $T=O(\epsilon^{-1/2})$ to plot the curves.
}\label{fig:Fig2} 
\end{figure*}

\noindent \textbf{\emph{Exponential error cancellation}}\\
Our analysis so far has not assumed a specific form for the scheduling function $u(x)$. Actually, a thoughtful choice of $u(x)$ can significantly reduce non-adiabatic errors by spending more time in regions with small energy gaps while passing quickly through regions with large energy gaps.  An intriguing question is whether optimizing the choice of the scheduling function $u(x)$ could lead to even better performance if we consider the algorithmic errors as well. Here, we provide an affirmative answer to this question.

% that the minimum gap $g_{\min}$ depends on $u(x)$. An intriguing question is whether optimizing the choice of the scheduling function $u(x)$ could lead to even better performance. Here, we provide an affirmative answer to this question.

Taking the GQSP method as an example, since it shows a greater improvement compared to the Trotterization method, we demonstrate that the infidelity can, in principle, be exponentially suppressed with $T$, in contrast to the $T^{-2}$ scaling in Theorem 3.

We first introduce the concept of the $Q$-th order path studied in~Ref.~\cite{hu2016optimizing}. 

\smallskip
\noindent\textbf{{Definition 1 ($Q$-th order path).}}
\emph{A scheduling function $u(x)$ is called as a $Q$-th order path as long as $u^{(q)}(0)=u^{(q)}(1)=0$ for every $1\leq q\leq Q$. If $u^{(1)}(0)\neq 0$ or $u^{(1)}(1)\neq0$, We call it a zeroth order path. Moreover, we define $\infty$ order path which satisfies $\lim_{T\to+\infty} u^{(q)}(0),u^{(q)}(1)=0$ for every $q\in\mathbb{N}^*$.}

\smallskip

\noindent We can then show that the infidelity has an improved asymptotic scaling with GQSP and the $Q$-th path.

\smallskip
\noindent
\textbf{{Theorem 4 (Exponential error cancellation with $Q$-th order path).}}
 \emph{For large $T$ and small $\delta t$, the infidelity $\mathcal I(1)$ in digital adiabatic evolution using GQSP scales as $O(\beta^2_{\text{GQSP}}T^{-2Q-2})$ if we choose the scheduling function $u(x)$ as the $Q$-th order path. Moreover, if we choose an $\infty$-order path, the infidelity scales as $O(\beta^2_{\text{GQSP}}e^{-T})$. Here $\beta_{\text{GQSP}}$ is a coefficient that depends on $K$, the path, and the Hamiltonian, but independent of $T$. }
 
 \smallskip

\noindent We note that while the infidelity is suppressed as $T^{-2Q-2}$ for large $Q$, the coefficient $\beta_{\text{GQSP}}$ introduces a factor that increases exponentially with $Q$. Consequently, there is generally a trade-off between these two terms, and $Q$ should be optimized accordingly. Additionally, $\beta_{\text{GQSP}}$ still depends on the inverse of the minimum energy gap $g_{\min}$ along the adiabatic path, meaning the adiabatic condition $T \gg g_{\text{min}}^{-2}$ cannot be bypassed even with a $Q$-th order path. However, for sufficiently large $T$ that satisfies $T \gg g_{\text{min}}^{-2}$, $\beta_{ \text{GQSP}}$ can be treated as a constant, thereby improving the infidelity scaling with the $Q$-th order path.
It is worth noting that similar conclusions have been drawn in the analysis of non-adiabatic errors~\cite{rezakhani2010accuracy, lin2020near}. Our results extend these findings by considering errors arising from Hamiltonian simulation algorithms.

% \Del{Since $\widetilde{U}$ in LCU with amplitude amplification~\cite{berry2014proceedings} is also diagonal under the basis $\{\phi_j(m/r)\}$~\cite{berry2015simulating} and one can still ensure $\Lambda_i(x)\geq 0$ and $\Delta_i(x)>0$ for a certain $K$, the scaling in Theorem 3 applies to LCU with amplitude amplification as well.}

\section{Applications}
Now we discuss the applications of our results. Generally, our findings can be applied to adiabatic quantum computation and adiabatic quantum state preparation. The proposed quantum algorithms make time-dependent simulation as easy as time-independent simulation. Furthermore, our infidelity bounds provide much tighter error analysis for the performance of digital adiabatic evolution, which can be leveraged to offer more precise resource estimates when solving practical problems.
Our theory suggests that for digital adiabatic evolution, the Trotter error does not accumulate with the total time $T$, and the GQSP error decreases as $T$ increases. This marks a significant improvement over the error bounds for Trotterization (as shown in Fig.~\ref{fig:Fig2}(\textbf{a})) and GQSP in general time-dependent Hamiltonian evolution processes.

More specifically, since the non-adiabaticity error scales as $O(T^{-2})$, achieving an infidelity within an error $\epsilon$ requires the total time $T$ to be at least $O(\epsilon^{-1/2})$. For the general Trotter error bound, the first-order Trotter error scales as $O(T^2\delta t^2)$, so $\delta t$ must be chosen to be at most $O(\epsilon)$, resulting in a query complexity that scales as $T/\delta t = O(\epsilon^{-3/2})$. Our theory, however, demonstrates that the first-order Trotter error scales as $O(\delta t^2)$, meaning that $\delta t = O(\epsilon^{1/2})$ is sufficient to achieve the same error $\epsilon$, reducing the query complexity to $O(\epsilon^{-1})$ (as shown in Fig.~\ref{fig:Fig2}(\textbf{c})). 

{To evaluate the improvement of our Trotter error bound over the previous one, we define an improvement factor $\eta$ as the ratio of the previous Trotter error bound to our proposed Trotter error bound and choose to study TFIM with a slight longitudinal field. By varying the longitudinal field strength $g$ from $0.01$ to $0.1$ we can change the minimum gap $g_{\text{min}}$ within the range [0.02, 0.31]. We plot $\eta$ as a function of $g_{\min}^{-1}$ for systems with $n=3$ to $6$ qubits in Fig.~\ref{fig:Fig2}(\textbf{b}). $T$ and $\delta t$ are chosen as the critical values required to ensure the adiabatic infidelity $\mathcal{I}_{\text{ad}}$ and simulation infidelity $\mathcal{I}_{\text{sim}}$ (rigorously defined in Eq.~(\ref{eq:infi})) are both below the threhold of $0.001$. The results indicate that the improvement factor increases with $g_{\min}^{-1}$ and reaches up to $10^6$ for $g_{\min}\approx 0.02$. This suggests that our Trotter error bound is significantly tighter than the previous one, especially for systems with small energy gaps.}

As for the GQSP method, the infidelity scales as $O(T^{-2})$ while query complexity scales as $\widetilde{O}(T\log(r/\delta))$. Although the query complexity is a slight improvement over the truncated Dyson series or Qubitization only when $\delta$ is much larger than $\epsilon$, the auxiliary qubits required are reduced from $\widetilde{O}(\log(T\delta t/\epsilon)$ $+$ $\log(L))$ to $O(\log(L))$ (as shown in Fig.~\ref{fig:Fig2}(\textbf{d})). 
Furthermore, the infidelity estimates can also serve as a target for optimizing parameter settings, thereby guiding the quantum circuit implementation under constrained circuit depth. For Trotterization, the optimal evolution time and time step are provided in Corollary 2.1.\\

Our results are applicable to practical problems, such as adiabatic Grover's algorithm~\cite{roland2002quantum}, the glued trees problem~\cite{somma2012quantum}, adiabatic state preparation for quantum chemistry problems~\cite{babbush2014adiabatic,sugisaki2022adiabatic}, and systems of linear  equations~\cite{subacsi2019quantum,lin2020near,costa2022optimal}. Here, we mainly focus on the last two applications to corroborate our theory. 
Here, we review the background of these two applications and show the application of our results in detail in the next section. 

In quantum chemistry, the second-quantized Hamiltonian of a molecular system is expressed as follows:
\begin{equation}
\hat{H}=\sum_{p,q=1}^N h_{pq} \hat{a}_p^{\dagger} \hat{a}_q+\frac{1}{2} \sum_{p,q,r,s=1}^N h_{psqr} \hat{a}_p^{\dagger} \hat{a}_q^{\dagger} \hat{a}_r \hat{a}_s,
\end{equation}
where $h_{pq}$ and $h_{pqrs}$ are the 1-electron and 2-electron repulsion integrals, and $\hat{a}_p^\dag$ and $\hat{a}_p$ are the creation and annihilation operators, respectively. This Hamiltonian represents the electronic Hamiltonian under the Born-Oppenheimer approximation~\cite{RevModPhys.92.015003, cao2019quantum}. To translate this Hamiltonian into the qubit representation, one can use the Jordan-Wigner or Bravyi-Kitaev transformations~\cite{RevModPhys.92.015003, cao2019quantum}.
Our goal is to determine the ground state of the Hamiltonian, which corresponds to the electronic structure of the molecule. While directly solving the Hamiltonian is generally a complex task, one potential approach is to utilize adiabatic evolution. In this method, $H_i$ is the Fock operator~\cite{hu2016optimizing} defined as
\begin{equation}
F=\sum_{i=1}^N h_{ii} \hat{a}_i^{\dagger} \hat{a}_i+\frac{1}{2} \sum_{i,j=1}^N (h_{ijji} \hat{a}_i^{\dagger} \hat{a}_j^{\dagger} \hat{a}_j \hat{a}_i + h_{ijij} \hat{a}_i^{\dagger} \hat{a}_j^{\dagger} \hat{a}_i \hat{a}_j)
\end{equation}
and $H_f = \hat{H}$. The process involves preparing the Hartree-Fock state and then adiabatically evolving from $H_i$ to $H_f$ to find the target ground state.

% The Hamiltonian above has the electron-electron interactions and will be set up as the target Hamiltonian in an adiabatic process. The initial Hamiltonian will be the Fock operator of the system, which has the Hartree-Fock state as its ground state. 

Another application is to use adiabatic evolution to solve systems of linear equations~\cite{subacsi2019quantum}. The goal in this case is to prepare a quantum state
\begin{equation}
    |x\rangle :=\frac{\sum_{j=1}^Nx_j|j\rangle}{\sqrt{\sum_{j=1}^N|x_j|^2}}=\frac{A^{-1}|b\rangle}{\|A^{-1}|b\rangle\|},
\end{equation}
where $\vec{x}=(x_1,\cdots,x_N)^T$ is the solution to the linear system $A\vec{x}=\vec{b}$. Here $A$ is an $N\times N$ Hermitian matrix and $\vec{b}$ is an $N$-dimensional normalized vector. For the linear system with a non-Hermitian matrix, we can use the dilation trick to reduce it to the Hermitian case~\cite{harrow2009quantum}. 
% \DA{We do not need the sparsity as long as we assume suitable input models for $A$ and $b$. But we need $A$ to be Hermitian. Maybe change to: Here $A$ is an $N$-by-$N$ Hermitian matrix and $\vec{b}$ is an $N$-dimensional normalized vector. For the linear system with non-Hermitian matrix, we can use the dilation trick to reduce it to the Hermitian case~\cite{HHLpaper}.} 
$|b\rangle \propto \sum_{j=1}^N b_j|j\rangle$ is the quantum state encoding of $\vec{b}$.
To solve this problem, we can consider the time-dependent Hamiltonian 
\begin{equation}\label{eq:ham_amplified}
H(s) = \sigma^+\otimes p(s)P_{\bar{b}}^{\perp}+\sigma^-\otimes P_{\bar{b}}^{\perp} p(s),
\end{equation}
where $\sigma^{\pm}=(X\pm iY)/2$ are single-qubit (raising and lowering) operators and $p(s) := (1-s)Z\otimes \mathbb{I}+sX\otimes A$, $P_{\bar{b}}^{\perp} := \mathbb{I}-|\bar{b}\rangle\langle \bar{b}|$,  $|\bar{b}\rangle := |+,b\rangle$.
One can easily verify that the eigenvalues of $H(s)$ are: 
\begin{equation*}
\{0, 0, \pm \sqrt{\lambda_1(s)},\cdots,\pm\sqrt{\lambda_{2N-1}(s)}\}.
\end{equation*}
Hence, the subspace of $H(s)$ corresponding to the zero eigenvalue is spanned by $|1,x(s)\rangle$ and $|0,\bar{b}\rangle$. Where we define:
\begin{equation}
    |x(s)\rangle :=\frac{A^{-1}(s)|\bar{b}\rangle}{\|A^{-1}(s)|\bar{b}\rangle\|}.
\end{equation}
If we start from the initial state $|1,x(0)\rangle=|1,\bar{b}\rangle$, the solution is encoded in the final state $|{1,x(1)}\rangle$. Since $\langle 0,\bar{b}|H(s)|1,x(s)\rangle=0$, there is no transition from $|1,x(s)\rangle$ to $|0,\bar{b}\rangle$. Therefore, degeneration in the ground state will not affect the performance of AQC.

% Our bound is much tighter than the general bound, which can make for more accurate resource estimations and assist in choosing the optimal time step $\delta t$ and total time $T$ for Trotterization and order $K$ for LCU in the above applications.

\section{Numerical Simulation} \label{sec:simulation_results}

In this section, we show numerical results for digital adiabatic evolution in solving the ground state preparation problem of the $N_2$ molecule and the problem of systems of linear equations. 
 
For the former case, we make a side-by-side comparison of the error behavior for Trotterization and GQSP in digital adiabatic evolution with previous research findings to show that the Trotter error is robust against increasing $T$ as given in Theorem 2, and the GQSP error can be suppressed by large $T$ as we prove in Theorem 3. We also show the relationship between the optimal $T$ and the circuit depth $d$ to validate Corollary 2.1.
For the latter case, we investigate the relationship between errors from GQSP with different orders $K$ and $\delta t$, demonstrating that Eq.~(\ref{eq:slice}) will not introduce the error term $\delta t$ as mentioned by \cite{van2001powerful}, further strengthening our conclusion that simulating time-dependent adiabatic evolution can be effectively achieved using only time-independent Hamiltonian simulation algorithms. We also choose scheduling functions with different order $Q$ to validate Theorem 4. These numerical results are in excellent agreement with our theory.

In our numerical test, we focus mainly on the overlap between $|\phi_0(1)\rangle$, $|\psi(1)\rangle$ and $|\Psi\rangle$, where $|\phi_0(1)\rangle$ is the ground state of $H_f$, $|\psi(1)\rangle$ is the state prepared by digital adiabatic evolution with Hamiltonian simulation errors, and $|\Psi\rangle$ is the state prepared by exact digital adiabatic evolution without Hamiltonian simulation errors. Then, we define three infidelities as follows.

\begin{equation}\label{eq:infi}
\begin{aligned}
&\mathcal{I}:= 1-|\langle \psi(1)| \phi_0(1) \rangle|^2,\\
&\mathcal{I}_{\text{ad}}:= 1-|\langle \phi_0(1)| \Psi \rangle|^2,\\
&\mathcal{I}_{\text{sim}}:= 1-|\langle \psi(1)| \Psi \rangle|^2,\\
\end{aligned}
\end{equation}
where 
$\mathcal{I}$, $\mathcal{I}_{\text{ad}}$ and $\mathcal{I}_{\text{sim}}$ correspond to the total infidelity, infidelity caused by non-adiabaticity, and infidelity caused by quantum simulation methods such as Trotterization and GQSP, respectively. 

Throughout this section, if not stated otherwise, we use the linear scheduling function $u(x)=x$.

\begin{table}[b]
\caption{\label{tab:schedule}%
Scheduling functions as different order paths used in this work.}
\begin{ruledtabular}
\begin{tabular}{cl}
Order $Q$ of the path & Scheduling function $u(x)$\\
\hline
0&$x$\\
1&$3x^2-2x^3$~\cite{hu2016optimizing}\\
2&$6x^5-15x^4+10x^3$~\cite{hu2016optimizing}\\
$\infty$&$C^{-1}\int_0^x \exp(-\frac{1}{x'(1-x')})dx'$\footnote{Here $C\int_0^1 \exp[-1/(x'-x'^2)]dx'$ is a normalization constant such that $u(1)=1$.}~\cite{lin2020near}\\
\end{tabular}
\end{ruledtabular}
\end{table}

\begin{figure*}
\includegraphics[width=1\textwidth]{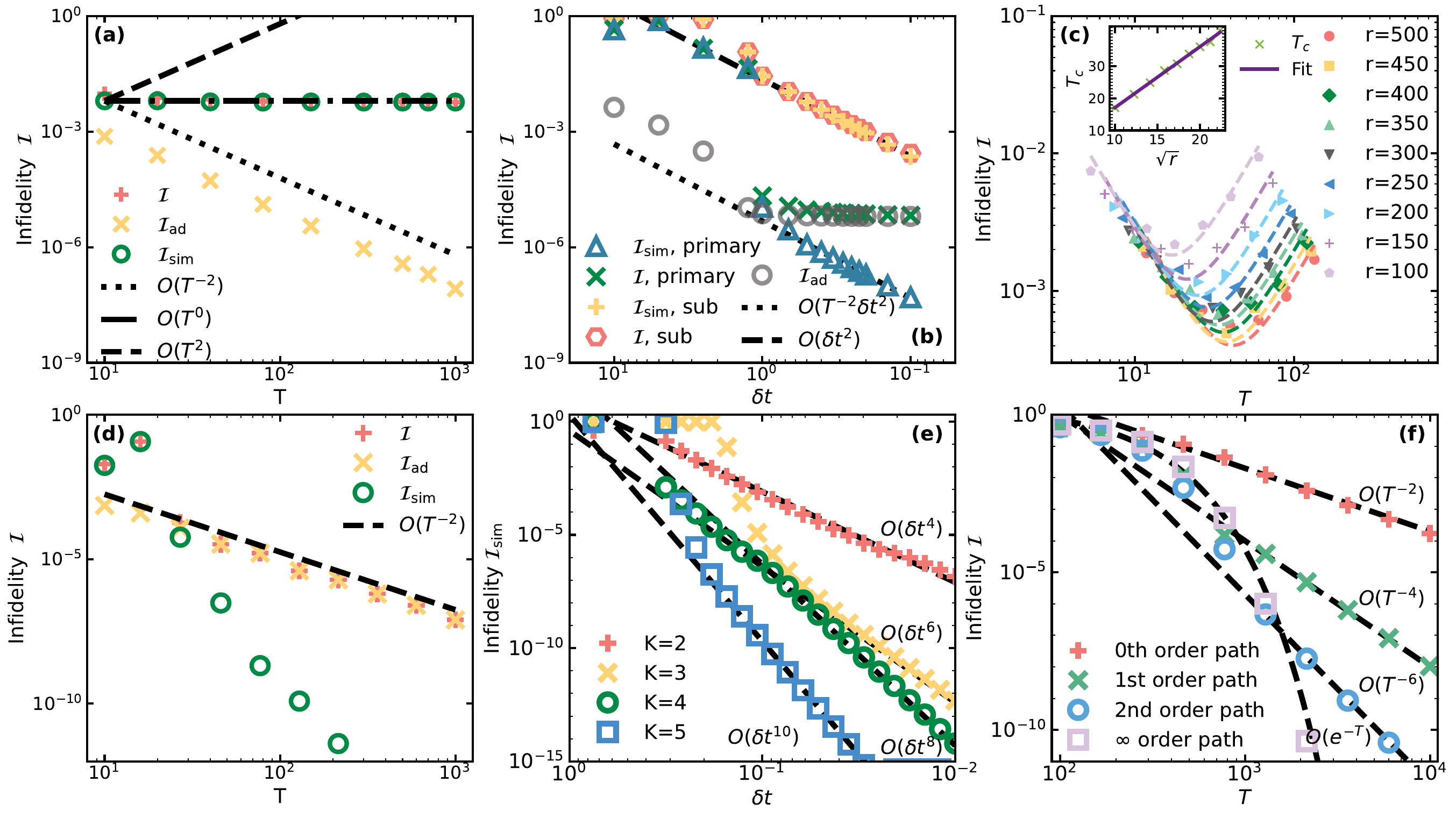}
\caption{ 
Numerical results for the $N_2$ molecule and systems of linear equations. (\textbf{a}) Infidelity of adiabatic state preparation (ASP) for the $N_2$ system with first-order sub-Trotterization at a fixed $\delta t=0.5$. The plot shows total infidelity $\mathcal{I}$, infidelity from non-adiabaticity $\mathcal{I}_{\text{ad}}$, and infidelity from Hamiltonian simulation $\mathcal{I}_{\text{sim}}$. (\textbf{b}) Infidelity of ASP for the $N_2$ system with primary Trotterization or first-order sub-Trotterization at a fixed $T=100$. (\textbf{c}) Parameter optimization in ASP for the $N_2$ system with first-order sub-Trotterization. The dashed curve represents a fit of the data using the model $\mathcal{I}=p(T/r)^2+bT^{-2}$. The minimum points $T_c$ of each fitting curve are shown with green crosses. The purple solid line is the linear fit of $T_c$ corresponding to $\sqrt{r}$. (\textbf{d}) Infidelity of ASP for the $N_2$ system at a fixed $\delta t=0.2$. Here we choose $K=\lfloor 4\log T \rfloor$. (\textbf{e}) Infidelity of adiabatic quantum computation (AQC) using GQSP for a linear equation system with orders $K=2$, $3$, $4$ and $5$ at a fixed $T=1000$. (\textbf{f}) Asymptotic behavior of AQC infidelity $\mathcal{I}$ using GQSP for a linear equation system at a fixed $\delta t=0.2$. Here we choose $K= 2\lfloor \log T \rfloor$. The scheduling function $u(x)$ is chosen as zeroth, first, second, and $\infty$ order paths, respectively. A scheduling function $u(x)$ is called a $Q$-th order path if $u^{(q)}(0)=u^{(q)}(1)=0$ for every $1\leq q\leq Q$. If $u^{(1)}(0)\neq 0$ or $u^{(1)}(1)\neq0$, it is a zeroth order path. An $\infty$ order path satisfies $\lim_{T\to+\infty} u^{(q)}(0),u^{(q)}(1)=0$ for every $q\in\mathbb{N}^*$.
}\label{Numerics}
\end{figure*}

\bigskip

\noindent \textbf{\emph{Results for the $N_2$ molecule}}\\
For the $N_2$ molecule, we consider the STO-3g basis set and restrict the Hilbert space into a complete active space~(CAS) of 6 electrons on 12 spin orbitals. As mentioned above, we use the Fock operator as the initial Hamiltonian $H_i$. So our evolution will start from $|\psi(0)\rangle=|\psi_{\text{HF}}\rangle$, where $|\psi_{\text{HF}}\rangle$ is the Hartree-Fock state. The final Hamiltonian $H_f$ is the Hamiltonian of the $N_2$ system at equilibrium bond length $1.134$ \AA.% 

First, we verify the $T$ scalings of our main theorems. In Fig.~\ref{Numerics}(a), we demonstrate that when applying sub-Trotterization at the first order, the total infidelity scales as $O(T^0)$ with $\delta t=0.5$, rather than accumulating or self-healing as suggested in Ref.~\cite{kovalsky2023self}. Both error contributions, the non-adiabatic error, and the Trotter error, are plotted. The non-adiabatic error shows an $O(T^{-2})$ scaling, while the Trotter error remains at the same magnitude. With $\delta t=0.5$, the Trotter error dominates the total error, leading to a breakdown of the self-healing effect.

In Fig.~\ref{Numerics}(b), we show that even with $\delta t=0.1$ and a relatively small $T=100$, the infidelity is still dominated by the Trotter error. Since the Trotter error scales as $O(\delta t^2)$, the total error decreases quadratically as $\delta t$ decreases. In the same range of $\delta t$, the error would be dominated by the non-adiabatic error if only primary Trotterization is considered, as also shown in the figure. This leads to different behaviors in the total error. Together with Fig.~\ref{Numerics}(a), this verifies the scaling in terms of $\delta t$ and $T$ as described in our Theorem 2.

As illustrated in Fig.~\ref{Numerics}(c), the infidelity $\mathcal{I}$ initially decreases but then increases as $T$ continues to increase. This behavior suggests that an optimal setting can be found, as we present in Corollary 2.1. We also fit the data using the model $\mathcal{I}=p(T/r)^2+bT^{-2}$ to identify the minimum points $T_c$ for different values of $r$. The linear fit of $T_c$ relative to $\sqrt{r}$ supports that $T_c=O(\sqrt{r})$.

For the GQSP method, numerically we choose $\alpha=12$, $\delta t=0.2$, and $K=\lfloor 4\log T \rfloor$. In Fig.~\ref{Numerics}(d), we observe that $I$ and $I_{\text{ad}}$ scale as $O(T^{-2})$ with increasing $T$. Besides, $\mathcal{I}_{\text{sim}}$ can also be bounded by $O(T^{-2})$, as shown in Fig.~\ref{Numerics}(d).  All this confirms the scaling behavior predicted in Theorem 3. 
\bigskip

\noindent \textbf{\emph{Results for systems of linear equations}}\\
Next, we turn our attention to the adiabatic solution of systems of linear equations to illustrate our findings in Theorem 3 and Theorem 4, particularly concerning the GQSP error and the boundary conditions for the scheduling function. By applying these theorems, we explore how the GQSP method behaves under different conditions, demonstrating the impact of the truncation order $K$ on the accuracy and efficiency of the adiabatic process, as well as the role of optimal scheduling functions in minimizing errors and enhancing performance. For this purpose, we generated a 4-sparse Hermitian matrix $A$ of dimension $N=16$ with the condition number $ \kappa \approx 10 $ (with an absolute error of $10^{-3}$) satisfying $\|A\|=1$.
% \DA{add the number of samples being selected?} 
We also generated a 4-sparse vector for $ \vec{b} $. According to Eq.(S37) in \cite{subacsi2019quantum}, we set $\alpha=10$ in GQSP.

In Fig.~\ref{Numerics}(e), we analyze how the simulation error $I_{\text{sim}}$ varies with respect to $\delta t$ for different truncation orders at a fixed $T=1000$, specifically $K=2, 3, 4, 5$. $I_{\text{sim}}$ exhibits an $O(\delta t^{2K})$ scaling, which is consistent with the truncation error of the GQSP method. This confirms that Eq.~(\ref{eq:slice}) does not introduce an additional error term $\delta t$, as previously suggested by \cite{van2001powerful}. 

In Fig.~\ref{Numerics}(f), we examine the asymptotic behavior of AQC infidelity $\mathcal{I}$ for a system of linear equations, using truncation order $K=2\lfloor \log T \rfloor$  and different orders of scheduling functions at a fixed $\delta t=0.2$. The scheduling functions tested are listed in TABLE \ref{tab:schedule}. The figure demonstrates that the infidelity scales as $O(T^{-2Q-2})$ for $Q=0, 1, 2$ and follows an $O(e^{-T})$ scaling for $Q=\infty$. These numerical results are consistent with the conclusions drawn in Theorem 4. It is worth noting that while $T$ is between 100-1000, amonge most of the cases shown in Fig.~\ref{Numerics}(f), 2nd
order path achieves the best performance, which could be attributed to the trade-off between the coefficient $\beta_{\text{GQSP}}$ and the infidelity scaling $O(T^{-2Q-2})$. The $\infty$ order path, while achieving the best infidelity scaling, has a larger coefficient $\beta_{\text{GQSP}}$, leading to a worse performance in this range of $T$.

\section{Discussion}

In this work, we have developed a comprehensive theory of error analysis in digital adiabatic evolution for general Hamiltonian simulation methods. Our analysis reveals that time-dependent adiabatic evolution can be realized via time-independent Hamiltonian simulation methods. Furthermore, in the context of Trotterization, there is a form of weak error cancellation, where the infidelity does not accumulate with the total evolution time $T$. In contrast, for advanced Hamiltonian simulation method GQSP, we observe strong error cancellation, with the infidelity consistently decreasing as the evolution time $T$ increases. Furthermore, we demonstrate that the total infidelity can decrease exponentially with $T$ when specific adiabatic paths are chosen, particularly by imposing certain boundary conditions on the derivatives of the scheduling function at the start and end of the evolution.
These findings present a promising pathway towards the development of more efficient and high-fidelity state preparation protocols using adiabatic evolution. Our results suggest that it is possible to significantly improve the performance of digital adiabatic evolution by carefully selecting and optimizing both the Hamiltonian simulation method and the scheduling function.
% Future work could explore whether more advanced Hamiltonian simulation techniques, such as quantum signal processing~\cite{low2017optimal} and qubitization~\cite{low2019hamiltonian}, can achieve even better scaling behavior, potentially surpassing the results obtained with Trotterization and GQSP. 

Future work could explore whether our findings can be combined with other remedial techniques for AQC, like counter-diabatic driving~\cite{finvzgar2025counterdiabatic}, fast quasi-adiabatic dynamics~\cite{wan2020fast} or adding a catalyst Hamiltonian~\cite{hormozi2017nonstoquastic}. These techniques are designed to mitigate non-adiabatic transitions during the evolution process. By integrating these methods with our error analysis framework, we could potentially develop more efficient adiabatic quantum algorithms that further reduce errors and improve overall performance. Additionally, integrating the state preparation strategies discussed in this work with other fault-tolerant quantum computation components, such as quantum phase estimation~\cite{preskill1998reliable}, could lead to a more robust framework for achieving practical quantum advantage. This integration would be particularly relevant in the context of the recent discussions on evaluating quantum advantage~\cite{lee2023evaluating}.
Moreover, future research could study the relation between error cancellation and error interference observed in Refs.~\cite{tranDestructiveErrorInterference2020, laydenFirstOrderTrotterError2022, yiSpectralAnalysisProduct2021} and study error-cancellation in more general Hamiltonian simulation tasks.
Lastly, our work focused on applying our results to solving the electronic structure problem and systems of linear equations. However, extending these results to other general quantum tasks presents intriguing opportunities for future research. Exploring how error-cancellation techniques can be utilized in diverse quantum algorithms, such as quantum materials, quantum machine learning, and quantum optimization, could significantly broaden the impact of our findings. Additionally, investigating the implications of our error-cancellation strategies in near-term quantum devices and their potential to improve the efficiency of quantum simulations in various scientific fields would be valuable avenues for further study.

\section{Acknowledgement}

The authors thank Lucas Kocia Kovalsky for insightful discussions on the self-healing effect. YL and XY are supported by the National Natural Science Foundation of China Grant (Grant No.~12361161602), NSAF (Grant No.~U2330201), and the Innovation Program for Quantum Science and Technology (Grant No.~2023ZD0300200). 
DA acknowledges the support by the Fundamental Research Funds for the Central Universities, Peking University. Q.Z. acknowledges funding from Innovation Program for Quantum Science and Technology via Project 2024ZD0301900, National Natural Science Foundation of China (NSFC) via Project No.~12347104 and No.~12305030, Guangdong Basic and Applied Basic Research Foundation via Project 2023A1515012185, Hong Kong Research Grant Council (RGC) via No.~27300823, N\_HKU718/23, and R6010-23, Guangdong Provincial Quantum Science Strategic Initiative No.~GDZX2303007, HKU Seed Fund for Basic Research for New Staff via Project 2201100596. The numerics is supported by the High-performance Computing Platform of Peking University.

\section{Code availability}
The Python source code for the the numerical simulation part is available via the github repository: \url{https://github.com/QuantyyLu/AQC_Simulation}.
% \section{Methods}

\bibliography{main}

\begin{thebibliography}{53}%
\makeatletter
\providecommand \@ifxundefined [1]{%
 \@ifx{#1\undefined}
}%
\providecommand \@ifnum [1]{%
 \ifnum #1\expandafter \@firstoftwo
 \else \expandafter \@secondoftwo
 \fi
}%
\providecommand \@ifx [1]{%
 \ifx #1\expandafter \@firstoftwo
 \else \expandafter \@secondoftwo
 \fi
}%
\providecommand \natexlab [1]{#1}%
\providecommand \enquote  [1]{``#1''}%
\providecommand \bibnamefont  [1]{#1}%
\providecommand \bibfnamefont [1]{#1}%
\providecommand \citenamefont [1]{#1}%
\providecommand \href@noop [0]{\@secondoftwo}%
\providecommand \href [0]{\begingroup \@sanitize@url \@href}%
\providecommand \@href[1]{\@@startlink{#1}\@@href}%
\providecommand \@@href[1]{\endgroup#1\@@endlink}%
\providecommand \@sanitize@url [0]{\catcode `\\12\catcode `\$12\catcode `\&12\catcode `\#12\catcode `\^12\catcode `\_12\catcode `\%12\relax}%
\providecommand \@@startlink[1]{}%
\providecommand \@@endlink[0]{}%
\providecommand \url  [0]{\begingroup\@sanitize@url \@url }%
\providecommand \@url [1]{\endgroup\@href {#1}{\urlprefix }}%
\providecommand \urlprefix  [0]{URL }%
\providecommand \Eprint [0]{\href }%
\providecommand \doibase [0]{https://doi.org/}%
\providecommand \selectlanguage [0]{\@gobble}%
\providecommand \bibinfo  [0]{\@secondoftwo}%
\providecommand \bibfield  [0]{\@secondoftwo}%
\providecommand \translation [1]{[#1]}%
\providecommand \BibitemOpen [0]{}%
\providecommand \bibitemStop [0]{}%
\providecommand \bibitemNoStop [0]{.\EOS\space}%
\providecommand \EOS [0]{\spacefactor3000\relax}%
\providecommand \BibitemShut  [1]{\csname bibitem#1\endcsname}%
\let\auto@bib@innerbib\@empty
%</preamble>
\bibitem [{\citenamefont {Abrams}\ and\ \citenamefont {Lloyd}(1999)}]{Abrams99}%
  \BibitemOpen
  \bibfield  {author} {\bibinfo {author} {\bibfnamefont {D.~S.}\ \bibnamefont {Abrams}}\ and\ \bibinfo {author} {\bibfnamefont {S.}~\bibnamefont {Lloyd}},\ }\bibfield  {title} {\bibinfo {title} {Quantum algorithm providing exponential speed increase for finding eigenvalues and eigenvectors},\ }\href {https://doi.org/10.1103/PhysRevLett.83.5162} {\bibfield  {journal} {\bibinfo  {journal} {Phys. Rev. Lett.}\ }\textbf {\bibinfo {volume} {83}},\ \bibinfo {pages} {5162} (\bibinfo {year} {1999})}\BibitemShut {NoStop}%
\bibitem [{\citenamefont {Aspuru-Guzik}\ \emph {et~al.}(2005)\citenamefont {Aspuru-Guzik}, \citenamefont {Dutoi}, \citenamefont {Love},\ and\ \citenamefont {Head-Gordon}}]{aspuru2005simulated}%
  \BibitemOpen
  \bibfield  {author} {\bibinfo {author} {\bibfnamefont {A.}~\bibnamefont {Aspuru-Guzik}}, \bibinfo {author} {\bibfnamefont {A.~D.}\ \bibnamefont {Dutoi}}, \bibinfo {author} {\bibfnamefont {P.~J.}\ \bibnamefont {Love}},\ and\ \bibinfo {author} {\bibfnamefont {M.}~\bibnamefont {Head-Gordon}},\ }\bibfield  {title} {\bibinfo {title} {Simulated quantum computation of molecular energies},\ }\href {https://doi.org/10.1126/science.1113479} {\bibfield  {journal} {\bibinfo  {journal} {Science}\ }\textbf {\bibinfo {volume} {309}},\ \bibinfo {pages} {1704} (\bibinfo {year} {2005})}\BibitemShut {NoStop}%
\bibitem [{\citenamefont {Lee}\ \emph {et~al.}(2023)\citenamefont {Lee}, \citenamefont {Lee}, \citenamefont {Zhai}, \citenamefont {Tong}, \citenamefont {Dalzell}, \citenamefont {Kumar}, \citenamefont {Helms}, \citenamefont {Gray}, \citenamefont {Cui}, \citenamefont {Liu} \emph {et~al.}}]{lee2023evaluating}%
  \BibitemOpen
  \bibfield  {author} {\bibinfo {author} {\bibfnamefont {S.}~\bibnamefont {Lee}}, \bibinfo {author} {\bibfnamefont {J.}~\bibnamefont {Lee}}, \bibinfo {author} {\bibfnamefont {H.}~\bibnamefont {Zhai}}, \bibinfo {author} {\bibfnamefont {Y.}~\bibnamefont {Tong}}, \bibinfo {author} {\bibfnamefont {A.~M.}\ \bibnamefont {Dalzell}}, \bibinfo {author} {\bibfnamefont {A.}~\bibnamefont {Kumar}}, \bibinfo {author} {\bibfnamefont {P.}~\bibnamefont {Helms}}, \bibinfo {author} {\bibfnamefont {J.}~\bibnamefont {Gray}}, \bibinfo {author} {\bibfnamefont {Z.-H.}\ \bibnamefont {Cui}}, \bibinfo {author} {\bibfnamefont {W.}~\bibnamefont {Liu}}, \emph {et~al.},\ }\bibfield  {title} {\bibinfo {title} {Evaluating the evidence for exponential quantum advantage in ground-state quantum chemistry},\ }\href@noop {} {\bibfield  {journal} {\bibinfo  {journal} {Nature communications}\ }\textbf {\bibinfo {volume} {14}},\ \bibinfo {pages} {1952} (\bibinfo {year} {2023})}\BibitemShut {NoStop}%
\bibitem [{\citenamefont {Farhi}\ \emph {et~al.}(2000)\citenamefont {Farhi}, \citenamefont {Goldstone}, \citenamefont {Gutmann},\ and\ \citenamefont {Sipser}}]{farhi2000quantum}%
  \BibitemOpen
  \bibfield  {author} {\bibinfo {author} {\bibfnamefont {E.}~\bibnamefont {Farhi}}, \bibinfo {author} {\bibfnamefont {J.}~\bibnamefont {Goldstone}}, \bibinfo {author} {\bibfnamefont {S.}~\bibnamefont {Gutmann}},\ and\ \bibinfo {author} {\bibfnamefont {M.}~\bibnamefont {Sipser}},\ }\bibfield  {title} {\bibinfo {title} {Quantum computation by adiabatic evolution},\ }\href@noop {} {\bibfield  {journal} {\bibinfo  {journal} {arXiv preprint quant-ph/0001106}\ } (\bibinfo {year} {2000})}\BibitemShut {NoStop}%
\bibitem [{\citenamefont {Kadowaki}\ and\ \citenamefont {Nishimori}(1998)}]{kadowaki1998quantum}%
  \BibitemOpen
  \bibfield  {author} {\bibinfo {author} {\bibfnamefont {T.}~\bibnamefont {Kadowaki}}\ and\ \bibinfo {author} {\bibfnamefont {H.}~\bibnamefont {Nishimori}},\ }\bibfield  {title} {\bibinfo {title} {Quantum annealing in the transverse ising model},\ }\href@noop {} {\bibfield  {journal} {\bibinfo  {journal} {Physical Review E}\ }\textbf {\bibinfo {volume} {58}},\ \bibinfo {pages} {5355} (\bibinfo {year} {1998})}\BibitemShut {NoStop}%
\bibitem [{\citenamefont {Santoro}\ and\ \citenamefont {Tosatti}(2006)}]{santoro2006optimization}%
  \BibitemOpen
  \bibfield  {author} {\bibinfo {author} {\bibfnamefont {G.~E.}\ \bibnamefont {Santoro}}\ and\ \bibinfo {author} {\bibfnamefont {E.}~\bibnamefont {Tosatti}},\ }\bibfield  {title} {\bibinfo {title} {Optimization using quantum mechanics: quantum annealing through adiabatic evolution},\ }\href@noop {} {\bibfield  {journal} {\bibinfo  {journal} {Journal of Physics A: Mathematical and General}\ }\textbf {\bibinfo {volume} {39}},\ \bibinfo {pages} {R393} (\bibinfo {year} {2006})}\BibitemShut {NoStop}%
\bibitem [{\citenamefont {Vandersypen}\ and\ \citenamefont {Chuang}(2004)}]{vandersypen2004nmr}%
  \BibitemOpen
  \bibfield  {author} {\bibinfo {author} {\bibfnamefont {L.~M.}\ \bibnamefont {Vandersypen}}\ and\ \bibinfo {author} {\bibfnamefont {I.~L.}\ \bibnamefont {Chuang}},\ }\bibfield  {title} {\bibinfo {title} {Nmr techniques for quantum control and computation},\ }\href@noop {} {\bibfield  {journal} {\bibinfo  {journal} {Reviews of modern physics}\ }\textbf {\bibinfo {volume} {76}},\ \bibinfo {pages} {1037} (\bibinfo {year} {2004})}\BibitemShut {NoStop}%
\bibitem [{\citenamefont {Zhou}\ \emph {et~al.}(2017)\citenamefont {Zhou}, \citenamefont {Baksic}, \citenamefont {Ribeiro}, \citenamefont {Yale}, \citenamefont {Heremans}, \citenamefont {Jerger}, \citenamefont {Auer}, \citenamefont {Burkard}, \citenamefont {Clerk},\ and\ \citenamefont {Awschalom}}]{zhou2017accelerated}%
  \BibitemOpen
  \bibfield  {author} {\bibinfo {author} {\bibfnamefont {B.~B.}\ \bibnamefont {Zhou}}, \bibinfo {author} {\bibfnamefont {A.}~\bibnamefont {Baksic}}, \bibinfo {author} {\bibfnamefont {H.}~\bibnamefont {Ribeiro}}, \bibinfo {author} {\bibfnamefont {C.~G.}\ \bibnamefont {Yale}}, \bibinfo {author} {\bibfnamefont {F.~J.}\ \bibnamefont {Heremans}}, \bibinfo {author} {\bibfnamefont {P.~C.}\ \bibnamefont {Jerger}}, \bibinfo {author} {\bibfnamefont {A.}~\bibnamefont {Auer}}, \bibinfo {author} {\bibfnamefont {G.}~\bibnamefont {Burkard}}, \bibinfo {author} {\bibfnamefont {A.~A.}\ \bibnamefont {Clerk}},\ and\ \bibinfo {author} {\bibfnamefont {D.~D.}\ \bibnamefont {Awschalom}},\ }\bibfield  {title} {\bibinfo {title} {Accelerated quantum control using superadiabatic dynamics in a solid-state lambda system},\ }\href@noop {} {\bibfield  {journal} {\bibinfo  {journal} {Nature Physics}\ }\textbf {\bibinfo {volume} {13}},\ \bibinfo {pages} {330} (\bibinfo {year} {2017})}\BibitemShut {NoStop}%
\bibitem [{\citenamefont {Van~Dam}\ \emph {et~al.}(2001)\citenamefont {Van~Dam}, \citenamefont {Mosca},\ and\ \citenamefont {Vazirani}}]{van2001powerful}%
  \BibitemOpen
  \bibfield  {author} {\bibinfo {author} {\bibfnamefont {W.}~\bibnamefont {Van~Dam}}, \bibinfo {author} {\bibfnamefont {M.}~\bibnamefont {Mosca}},\ and\ \bibinfo {author} {\bibfnamefont {U.}~\bibnamefont {Vazirani}},\ }\bibfield  {title} {\bibinfo {title} {How powerful is adiabatic quantum computation?},\ }in\ \href@noop {} {\emph {\bibinfo {booktitle} {Proceedings 42nd IEEE symposium on foundations of computer science}}}\ (\bibinfo {organization} {IEEE},\ \bibinfo {year} {2001})\ pp.\ \bibinfo {pages} {279--287}\BibitemShut {NoStop}%
\bibitem [{\citenamefont {Roland}\ and\ \citenamefont {Cerf}(2002)}]{roland2002quantum}%
  \BibitemOpen
  \bibfield  {author} {\bibinfo {author} {\bibfnamefont {J.}~\bibnamefont {Roland}}\ and\ \bibinfo {author} {\bibfnamefont {N.~J.}\ \bibnamefont {Cerf}},\ }\bibfield  {title} {\bibinfo {title} {Quantum search by local adiabatic evolution},\ }\href@noop {} {\bibfield  {journal} {\bibinfo  {journal} {Physical Review A}\ }\textbf {\bibinfo {volume} {65}},\ \bibinfo {pages} {042308} (\bibinfo {year} {2002})}\BibitemShut {NoStop}%
\bibitem [{\citenamefont {Somma}\ \emph {et~al.}(2012)\citenamefont {Somma}, \citenamefont {Nagaj},\ and\ \citenamefont {Kieferov{\'a}}}]{somma2012quantum}%
  \BibitemOpen
  \bibfield  {author} {\bibinfo {author} {\bibfnamefont {R.~D.}\ \bibnamefont {Somma}}, \bibinfo {author} {\bibfnamefont {D.}~\bibnamefont {Nagaj}},\ and\ \bibinfo {author} {\bibfnamefont {M.}~\bibnamefont {Kieferov{\'a}}},\ }\bibfield  {title} {\bibinfo {title} {Quantum speedup by quantum annealing},\ }\href@noop {} {\bibfield  {journal} {\bibinfo  {journal} {Physical review letters}\ }\textbf {\bibinfo {volume} {109}},\ \bibinfo {pages} {050501} (\bibinfo {year} {2012})}\BibitemShut {NoStop}%
\bibitem [{\citenamefont {Garnerone}\ \emph {et~al.}(2012)\citenamefont {Garnerone}, \citenamefont {Zanardi},\ and\ \citenamefont {Lidar}}]{garnerone2012adiabatic}%
  \BibitemOpen
  \bibfield  {author} {\bibinfo {author} {\bibfnamefont {S.}~\bibnamefont {Garnerone}}, \bibinfo {author} {\bibfnamefont {P.}~\bibnamefont {Zanardi}},\ and\ \bibinfo {author} {\bibfnamefont {D.~A.}\ \bibnamefont {Lidar}},\ }\bibfield  {title} {\bibinfo {title} {Adiabatic quantum algorithm for search engine ranking},\ }\href@noop {} {\bibfield  {journal} {\bibinfo  {journal} {Physical review letters}\ }\textbf {\bibinfo {volume} {108}},\ \bibinfo {pages} {230506} (\bibinfo {year} {2012})}\BibitemShut {NoStop}%
\bibitem [{\citenamefont {McArdle}\ \emph {et~al.}(2020)\citenamefont {McArdle}, \citenamefont {Endo}, \citenamefont {Aspuru-Guzik}, \citenamefont {Benjamin},\ and\ \citenamefont {Yuan}}]{RevModPhys.92.015003}%
  \BibitemOpen
  \bibfield  {author} {\bibinfo {author} {\bibfnamefont {S.}~\bibnamefont {McArdle}}, \bibinfo {author} {\bibfnamefont {S.}~\bibnamefont {Endo}}, \bibinfo {author} {\bibfnamefont {A.}~\bibnamefont {Aspuru-Guzik}}, \bibinfo {author} {\bibfnamefont {S.~C.}\ \bibnamefont {Benjamin}},\ and\ \bibinfo {author} {\bibfnamefont {X.}~\bibnamefont {Yuan}},\ }\bibfield  {title} {\bibinfo {title} {Quantum computational chemistry},\ }\href {https://doi.org/10.1103/RevModPhys.92.015003} {\bibfield  {journal} {\bibinfo  {journal} {Rev. Mod. Phys.}\ }\textbf {\bibinfo {volume} {92}},\ \bibinfo {pages} {015003} (\bibinfo {year} {2020})}\BibitemShut {NoStop}%
\bibitem [{\citenamefont {Cao}\ \emph {et~al.}(2019)\citenamefont {Cao}, \citenamefont {Romero}, \citenamefont {Olson}, \citenamefont {Degroote}, \citenamefont {Johnson}, \citenamefont {Kieferov{\'a}}, \citenamefont {Kivlichan}, \citenamefont {Menke}, \citenamefont {Peropadre}, \citenamefont {Sawaya} \emph {et~al.}}]{cao2019quantum}%
  \BibitemOpen
  \bibfield  {author} {\bibinfo {author} {\bibfnamefont {Y.}~\bibnamefont {Cao}}, \bibinfo {author} {\bibfnamefont {J.}~\bibnamefont {Romero}}, \bibinfo {author} {\bibfnamefont {J.~P.}\ \bibnamefont {Olson}}, \bibinfo {author} {\bibfnamefont {M.}~\bibnamefont {Degroote}}, \bibinfo {author} {\bibfnamefont {P.~D.}\ \bibnamefont {Johnson}}, \bibinfo {author} {\bibfnamefont {M.}~\bibnamefont {Kieferov{\'a}}}, \bibinfo {author} {\bibfnamefont {I.~D.}\ \bibnamefont {Kivlichan}}, \bibinfo {author} {\bibfnamefont {T.}~\bibnamefont {Menke}}, \bibinfo {author} {\bibfnamefont {B.}~\bibnamefont {Peropadre}}, \bibinfo {author} {\bibfnamefont {N.~P.}\ \bibnamefont {Sawaya}}, \emph {et~al.},\ }\bibfield  {title} {\bibinfo {title} {Quantum chemistry in the age of quantum computing},\ }\href@noop {} {\bibfield  {journal} {\bibinfo  {journal} {Chemical reviews}\ }\textbf {\bibinfo {volume} {119}},\ \bibinfo {pages} {10856} (\bibinfo {year} {2019})}\BibitemShut {NoStop}%
\bibitem [{\citenamefont {Bauer}\ \emph {et~al.}(2020)\citenamefont {Bauer}, \citenamefont {Bravyi}, \citenamefont {Motta},\ and\ \citenamefont {Chan}}]{Bauer_2020}%
  \BibitemOpen
  \bibfield  {author} {\bibinfo {author} {\bibfnamefont {B.}~\bibnamefont {Bauer}}, \bibinfo {author} {\bibfnamefont {S.}~\bibnamefont {Bravyi}}, \bibinfo {author} {\bibfnamefont {M.}~\bibnamefont {Motta}},\ and\ \bibinfo {author} {\bibfnamefont {G.~K.-L.}\ \bibnamefont {Chan}},\ }\bibfield  {title} {\bibinfo {title} {Quantum algorithms for quantum chemistry and quantum materials science},\ }\href {https://doi.org/10.1021/acs.chemrev.9b00829} {\bibfield  {journal} {\bibinfo  {journal} {Chemical Reviews}\ }\textbf {\bibinfo {volume} {120}},\ \bibinfo {pages} {12685–12717} (\bibinfo {year} {2020})}\BibitemShut {NoStop}%
\bibitem [{\citenamefont {Kempe}\ \emph {et~al.}(2006)\citenamefont {Kempe}, \citenamefont {Kitaev},\ and\ \citenamefont {Regev}}]{kempe2006complexity}%
  \BibitemOpen
  \bibfield  {author} {\bibinfo {author} {\bibfnamefont {J.}~\bibnamefont {Kempe}}, \bibinfo {author} {\bibfnamefont {A.}~\bibnamefont {Kitaev}},\ and\ \bibinfo {author} {\bibfnamefont {O.}~\bibnamefont {Regev}},\ }\bibfield  {title} {\bibinfo {title} {The complexity of the local hamiltonian problem},\ }\href@noop {} {\bibfield  {journal} {\bibinfo  {journal} {Siam journal on computing}\ }\textbf {\bibinfo {volume} {35}},\ \bibinfo {pages} {1070} (\bibinfo {year} {2006})}\BibitemShut {NoStop}%
\bibitem [{\citenamefont {Albash}\ and\ \citenamefont {Lidar}(2018)}]{RevModPhys.90.015002}%
  \BibitemOpen
  \bibfield  {author} {\bibinfo {author} {\bibfnamefont {T.}~\bibnamefont {Albash}}\ and\ \bibinfo {author} {\bibfnamefont {D.~A.}\ \bibnamefont {Lidar}},\ }\bibfield  {title} {\bibinfo {title} {Adiabatic quantum computation},\ }\href {https://doi.org/10.1103/RevModPhys.90.015002} {\bibfield  {journal} {\bibinfo  {journal} {Rev. Mod. Phys.}\ }\textbf {\bibinfo {volume} {90}},\ \bibinfo {pages} {015002} (\bibinfo {year} {2018})}\BibitemShut {NoStop}%
\bibitem [{\citenamefont {Nielsen}\ and\ \citenamefont {Chuang}(2010)}]{nielsen2010quantum}%
  \BibitemOpen
  \bibfield  {author} {\bibinfo {author} {\bibfnamefont {M.~A.}\ \bibnamefont {Nielsen}}\ and\ \bibinfo {author} {\bibfnamefont {I.~L.}\ \bibnamefont {Chuang}},\ }\href@noop {} {\emph {\bibinfo {title} {Quantum computation and quantum information}}}\ (\bibinfo  {publisher} {Cambridge university press},\ \bibinfo {year} {2010})\BibitemShut {NoStop}%
\bibitem [{\citenamefont {Childs}\ \emph {et~al.}(2021)\citenamefont {Childs}, \citenamefont {Su}, \citenamefont {Tran}, \citenamefont {Wiebe},\ and\ \citenamefont {Zhu}}]{childs2021theory}%
  \BibitemOpen
  \bibfield  {author} {\bibinfo {author} {\bibfnamefont {A.~M.}\ \bibnamefont {Childs}}, \bibinfo {author} {\bibfnamefont {Y.}~\bibnamefont {Su}}, \bibinfo {author} {\bibfnamefont {M.~C.}\ \bibnamefont {Tran}}, \bibinfo {author} {\bibfnamefont {N.}~\bibnamefont {Wiebe}},\ and\ \bibinfo {author} {\bibfnamefont {S.}~\bibnamefont {Zhu}},\ }\bibfield  {title} {\bibinfo {title} {Theory of trotter error with commutator scaling},\ }\href {https://doi.org/10.1103/PhysRevX.11.011020} {\bibfield  {journal} {\bibinfo  {journal} {Phys. Rev. X}\ }\textbf {\bibinfo {volume} {11}},\ \bibinfo {pages} {011020} (\bibinfo {year} {2021})}\BibitemShut {NoStop}%
\bibitem [{\citenamefont {Childs}\ \emph {et~al.}(2018)\citenamefont {Childs}, \citenamefont {Maslov}, \citenamefont {Nam}, \citenamefont {Ross},\ and\ \citenamefont {Su}}]{childs2018toward}%
  \BibitemOpen
  \bibfield  {author} {\bibinfo {author} {\bibfnamefont {A.~M.}\ \bibnamefont {Childs}}, \bibinfo {author} {\bibfnamefont {D.}~\bibnamefont {Maslov}}, \bibinfo {author} {\bibfnamefont {Y.}~\bibnamefont {Nam}}, \bibinfo {author} {\bibfnamefont {N.~J.}\ \bibnamefont {Ross}},\ and\ \bibinfo {author} {\bibfnamefont {Y.}~\bibnamefont {Su}},\ }\bibfield  {title} {\bibinfo {title} {Toward the first quantum simulation with quantum speedup},\ }\href@noop {} {\bibfield  {journal} {\bibinfo  {journal} {Proceedings of the National Academy of Sciences}\ }\textbf {\bibinfo {volume} {115}},\ \bibinfo {pages} {9456} (\bibinfo {year} {2018})}\BibitemShut {NoStop}%
\bibitem [{\citenamefont {Jansen}\ \emph {et~al.}(2007)\citenamefont {Jansen}, \citenamefont {Ruskai},\ and\ \citenamefont {Seiler}}]{jansen2007bounds}%
  \BibitemOpen
  \bibfield  {author} {\bibinfo {author} {\bibfnamefont {S.}~\bibnamefont {Jansen}}, \bibinfo {author} {\bibfnamefont {M.-B.}\ \bibnamefont {Ruskai}},\ and\ \bibinfo {author} {\bibfnamefont {R.}~\bibnamefont {Seiler}},\ }\bibfield  {title} {\bibinfo {title} {Bounds for the adiabatic approximation with applications to quantum computation},\ }\href@noop {} {\bibfield  {journal} {\bibinfo  {journal} {Journal of Mathematical Physics}\ }\textbf {\bibinfo {volume} {48}} (\bibinfo {year} {2007})}\BibitemShut {NoStop}%
\bibitem [{\citenamefont {Elgart}\ and\ \citenamefont {Hagedorn}(2012)}]{elgart2012note}%
  \BibitemOpen
  \bibfield  {author} {\bibinfo {author} {\bibfnamefont {A.}~\bibnamefont {Elgart}}\ and\ \bibinfo {author} {\bibfnamefont {G.~A.}\ \bibnamefont {Hagedorn}},\ }\bibfield  {title} {\bibinfo {title} {A note on the switching adiabatic theorem},\ }\href@noop {} {\bibfield  {journal} {\bibinfo  {journal} {Journal of Mathematical Physics}\ }\textbf {\bibinfo {volume} {53}} (\bibinfo {year} {2012})}\BibitemShut {NoStop}%
\bibitem [{\citenamefont {Motlagh}\ and\ \citenamefont {Wiebe}(2024)}]{motlagh2024generalized}%
  \BibitemOpen
  \bibfield  {author} {\bibinfo {author} {\bibfnamefont {D.}~\bibnamefont {Motlagh}}\ and\ \bibinfo {author} {\bibfnamefont {N.}~\bibnamefont {Wiebe}},\ }\bibfield  {title} {\bibinfo {title} {Generalized quantum signal processing},\ }\href@noop {} {\bibfield  {journal} {\bibinfo  {journal} {PRX Quantum}\ }\textbf {\bibinfo {volume} {5}},\ \bibinfo {pages} {020368} (\bibinfo {year} {2024})}\BibitemShut {NoStop}%
\bibitem [{\citenamefont {Kieferov{\'a}}\ \emph {et~al.}(2019)\citenamefont {Kieferov{\'a}}, \citenamefont {Scherer},\ and\ \citenamefont {Berry}}]{kieferova2019simulating}%
  \BibitemOpen
  \bibfield  {author} {\bibinfo {author} {\bibfnamefont {M.}~\bibnamefont {Kieferov{\'a}}}, \bibinfo {author} {\bibfnamefont {A.}~\bibnamefont {Scherer}},\ and\ \bibinfo {author} {\bibfnamefont {D.~W.}\ \bibnamefont {Berry}},\ }\bibfield  {title} {\bibinfo {title} {Simulating the dynamics of time-dependent hamiltonians with a truncated dyson series},\ }\href@noop {} {\bibfield  {journal} {\bibinfo  {journal} {Physical Review A}\ }\textbf {\bibinfo {volume} {99}},\ \bibinfo {pages} {042314} (\bibinfo {year} {2019})}\BibitemShut {NoStop}%
\bibitem [{\citenamefont {Watkins}\ \emph {et~al.}(2024)\citenamefont {Watkins}, \citenamefont {Wiebe}, \citenamefont {Roggero},\ and\ \citenamefont {Lee}}]{watkins2024time}%
  \BibitemOpen
  \bibfield  {author} {\bibinfo {author} {\bibfnamefont {J.}~\bibnamefont {Watkins}}, \bibinfo {author} {\bibfnamefont {N.}~\bibnamefont {Wiebe}}, \bibinfo {author} {\bibfnamefont {A.}~\bibnamefont {Roggero}},\ and\ \bibinfo {author} {\bibfnamefont {D.}~\bibnamefont {Lee}},\ }\bibfield  {title} {\bibinfo {title} {Time-dependent hamiltonian simulation using discrete-clock constructions},\ }\href@noop {} {\bibfield  {journal} {\bibinfo  {journal} {PRX Quantum}\ }\textbf {\bibinfo {volume} {5}},\ \bibinfo {pages} {040316} (\bibinfo {year} {2024})}\BibitemShut {NoStop}%
\bibitem [{\citenamefont {Suba{\c{s}}{\i}}\ \emph {et~al.}(2019)\citenamefont {Suba{\c{s}}{\i}}, \citenamefont {Somma},\ and\ \citenamefont {Orsucci}}]{subacsi2019quantum}%
  \BibitemOpen
  \bibfield  {author} {\bibinfo {author} {\bibfnamefont {Y.}~\bibnamefont {Suba{\c{s}}{\i}}}, \bibinfo {author} {\bibfnamefont {R.~D.}\ \bibnamefont {Somma}},\ and\ \bibinfo {author} {\bibfnamefont {D.}~\bibnamefont {Orsucci}},\ }\bibfield  {title} {\bibinfo {title} {Quantum algorithms for systems of linear equations inspired by adiabatic quantum computing},\ }\href@noop {} {\bibfield  {journal} {\bibinfo  {journal} {Physical review letters}\ }\textbf {\bibinfo {volume} {122}},\ \bibinfo {pages} {060504} (\bibinfo {year} {2019})}\BibitemShut {NoStop}%
\bibitem [{\citenamefont {Preskill}(1998)}]{preskill1998reliable}%
  \BibitemOpen
  \bibfield  {author} {\bibinfo {author} {\bibfnamefont {J.}~\bibnamefont {Preskill}},\ }\bibfield  {title} {\bibinfo {title} {Reliable quantum computers},\ }\href@noop {} {\bibfield  {journal} {\bibinfo  {journal} {Proceedings of the Royal Society of London. Series A: Mathematical, Physical and Engineering Sciences}\ }\textbf {\bibinfo {volume} {454}},\ \bibinfo {pages} {385} (\bibinfo {year} {1998})}\BibitemShut {NoStop}%
\bibitem [{\citenamefont {Preskill}(2018)}]{preskill2018quantum}%
  \BibitemOpen
  \bibfield  {author} {\bibinfo {author} {\bibfnamefont {J.}~\bibnamefont {Preskill}},\ }\bibfield  {title} {\bibinfo {title} {Quantum computing in the nisq era and beyond},\ }\href@noop {} {\bibfield  {journal} {\bibinfo  {journal} {Quantum}\ }\textbf {\bibinfo {volume} {2}},\ \bibinfo {pages} {79} (\bibinfo {year} {2018})}\BibitemShut {NoStop}%
\bibitem [{\citenamefont {Degani}\ and\ \citenamefont {Schiff}(2006)}]{degani2006rcms}%
  \BibitemOpen
  \bibfield  {author} {\bibinfo {author} {\bibfnamefont {I.}~\bibnamefont {Degani}}\ and\ \bibinfo {author} {\bibfnamefont {J.}~\bibnamefont {Schiff}},\ }\bibfield  {title} {\bibinfo {title} {Rcms: Right correction magnus series approach for oscillatory odes},\ }\href@noop {} {\bibfield  {journal} {\bibinfo  {journal} {Journal of computational and applied mathematics}\ }\textbf {\bibinfo {volume} {193}},\ \bibinfo {pages} {413} (\bibinfo {year} {2006})}\BibitemShut {NoStop}%
\bibitem [{\citenamefont {Iserles}\ and\ \citenamefont {N{\o}rsett}(2005)}]{iserles2005efficient}%
  \BibitemOpen
  \bibfield  {author} {\bibinfo {author} {\bibfnamefont {A.}~\bibnamefont {Iserles}}\ and\ \bibinfo {author} {\bibfnamefont {S.~P.}\ \bibnamefont {N{\o}rsett}},\ }\bibfield  {title} {\bibinfo {title} {Efficient quadrature of highly oscillatory integrals using derivatives},\ }\href@noop {} {\bibfield  {journal} {\bibinfo  {journal} {Proceedings of the Royal Society A: Mathematical, Physical and Engineering Sciences}\ }\textbf {\bibinfo {volume} {461}},\ \bibinfo {pages} {1383} (\bibinfo {year} {2005})}\BibitemShut {NoStop}%
\bibitem [{\citenamefont {Suzuki}(1991)}]{suzukiGeneralTheoryFractal1991}%
  \BibitemOpen
  \bibfield  {author} {\bibinfo {author} {\bibfnamefont {M.}~\bibnamefont {Suzuki}},\ }\bibfield  {title} {\bibinfo {title} {General theory of fractal path integrals with applications to many-body theories and statistical physics},\ }\href {https://doi.org/10.1063/1.529425} {\bibfield  {journal} {\bibinfo  {journal} {Journal of Mathematical Physics}\ }\textbf {\bibinfo {volume} {32}},\ \bibinfo {pages} {400} (\bibinfo {year} {1991})}\BibitemShut {NoStop}%
\bibitem [{\citenamefont {Blatt}\ and\ \citenamefont {Roos}(2012)}]{blatt2012quantum}%
  \BibitemOpen
  \bibfield  {author} {\bibinfo {author} {\bibfnamefont {R.}~\bibnamefont {Blatt}}\ and\ \bibinfo {author} {\bibfnamefont {C.~F.}\ \bibnamefont {Roos}},\ }\bibfield  {title} {\bibinfo {title} {Quantum simulations with trapped ions},\ }\href@noop {} {\bibfield  {journal} {\bibinfo  {journal} {Nature Physics}\ }\textbf {\bibinfo {volume} {8}},\ \bibinfo {pages} {277} (\bibinfo {year} {2012})}\BibitemShut {NoStop}%
\bibitem [{\citenamefont {Gross}\ and\ \citenamefont {Bloch}(2017)}]{gross2017quantum}%
  \BibitemOpen
  \bibfield  {author} {\bibinfo {author} {\bibfnamefont {C.}~\bibnamefont {Gross}}\ and\ \bibinfo {author} {\bibfnamefont {I.}~\bibnamefont {Bloch}},\ }\bibfield  {title} {\bibinfo {title} {Quantum simulations with ultracold atoms in optical lattices},\ }\href@noop {} {\bibfield  {journal} {\bibinfo  {journal} {Science}\ }\textbf {\bibinfo {volume} {357}},\ \bibinfo {pages} {995} (\bibinfo {year} {2017})}\BibitemShut {NoStop}%
\bibitem [{\citenamefont {Aspuru-Guzik}\ and\ \citenamefont {Walther}(2012)}]{aspuru2012photonic}%
  \BibitemOpen
  \bibfield  {author} {\bibinfo {author} {\bibfnamefont {A.}~\bibnamefont {Aspuru-Guzik}}\ and\ \bibinfo {author} {\bibfnamefont {P.}~\bibnamefont {Walther}},\ }\bibfield  {title} {\bibinfo {title} {Photonic quantum simulators},\ }\href@noop {} {\bibfield  {journal} {\bibinfo  {journal} {Nature physics}\ }\textbf {\bibinfo {volume} {8}},\ \bibinfo {pages} {285} (\bibinfo {year} {2012})}\BibitemShut {NoStop}%
\bibitem [{\citenamefont {Daley}\ \emph {et~al.}(2022)\citenamefont {Daley}, \citenamefont {Bloch}, \citenamefont {Kokail}, \citenamefont {Flannigan}, \citenamefont {Pearson}, \citenamefont {Troyer},\ and\ \citenamefont {Zoller}}]{daley2022practical}%
  \BibitemOpen
  \bibfield  {author} {\bibinfo {author} {\bibfnamefont {A.~J.}\ \bibnamefont {Daley}}, \bibinfo {author} {\bibfnamefont {I.}~\bibnamefont {Bloch}}, \bibinfo {author} {\bibfnamefont {C.}~\bibnamefont {Kokail}}, \bibinfo {author} {\bibfnamefont {S.}~\bibnamefont {Flannigan}}, \bibinfo {author} {\bibfnamefont {N.}~\bibnamefont {Pearson}}, \bibinfo {author} {\bibfnamefont {M.}~\bibnamefont {Troyer}},\ and\ \bibinfo {author} {\bibfnamefont {P.}~\bibnamefont {Zoller}},\ }\bibfield  {title} {\bibinfo {title} {Practical quantum advantage in quantum simulation},\ }\href@noop {} {\bibfield  {journal} {\bibinfo  {journal} {Nature}\ }\textbf {\bibinfo {volume} {607}},\ \bibinfo {pages} {667} (\bibinfo {year} {2022})}\BibitemShut {NoStop}%
\bibitem [{\citenamefont {Kovalsky}\ \emph {et~al.}(2023)\citenamefont {Kovalsky}, \citenamefont {Calderon-Vargas}, \citenamefont {Grace}, \citenamefont {Magann}, \citenamefont {Larsen}, \citenamefont {Baczewski},\ and\ \citenamefont {Sarovar}}]{kovalsky2023self}%
  \BibitemOpen
  \bibfield  {author} {\bibinfo {author} {\bibfnamefont {L.~K.}\ \bibnamefont {Kovalsky}}, \bibinfo {author} {\bibfnamefont {F.~A.}\ \bibnamefont {Calderon-Vargas}}, \bibinfo {author} {\bibfnamefont {M.~D.}\ \bibnamefont {Grace}}, \bibinfo {author} {\bibfnamefont {A.~B.}\ \bibnamefont {Magann}}, \bibinfo {author} {\bibfnamefont {J.~B.}\ \bibnamefont {Larsen}}, \bibinfo {author} {\bibfnamefont {A.~D.}\ \bibnamefont {Baczewski}},\ and\ \bibinfo {author} {\bibfnamefont {M.}~\bibnamefont {Sarovar}},\ }\bibfield  {title} {\bibinfo {title} {Self-healing of trotter error in digital adiabatic state preparation},\ }\href@noop {} {\bibfield  {journal} {\bibinfo  {journal} {Physical Review Letters}\ }\textbf {\bibinfo {volume} {131}},\ \bibinfo {pages} {060602} (\bibinfo {year} {2023})}\BibitemShut {NoStop}%
\bibitem [{\citenamefont {Yi}(2021)}]{yi2021success}%
  \BibitemOpen
  \bibfield  {author} {\bibinfo {author} {\bibfnamefont {C.}~\bibnamefont {Yi}},\ }\bibfield  {title} {\bibinfo {title} {Success of digital adiabatic simulation with large trotter step},\ }\href@noop {} {\bibfield  {journal} {\bibinfo  {journal} {Physical Review A}\ }\textbf {\bibinfo {volume} {104}},\ \bibinfo {pages} {052603} (\bibinfo {year} {2021})}\BibitemShut {NoStop}%
\bibitem [{\citenamefont {Layden}(2022{\natexlab{a}})}]{layden2022first}%
  \BibitemOpen
  \bibfield  {author} {\bibinfo {author} {\bibfnamefont {D.}~\bibnamefont {Layden}},\ }\bibfield  {title} {\bibinfo {title} {First-order trotter error from a second-order perspective},\ }\href@noop {} {\bibfield  {journal} {\bibinfo  {journal} {Physical Review Letters}\ }\textbf {\bibinfo {volume} {128}},\ \bibinfo {pages} {210501} (\bibinfo {year} {2022}{\natexlab{a}})}\BibitemShut {NoStop}%
\bibitem [{\citenamefont {Sugisaki}\ \emph {et~al.}(2022)\citenamefont {Sugisaki}, \citenamefont {Toyota}, \citenamefont {Sato}, \citenamefont {Shiomi},\ and\ \citenamefont {Takui}}]{sugisaki2022adiabatic}%
  \BibitemOpen
  \bibfield  {author} {\bibinfo {author} {\bibfnamefont {K.}~\bibnamefont {Sugisaki}}, \bibinfo {author} {\bibfnamefont {K.}~\bibnamefont {Toyota}}, \bibinfo {author} {\bibfnamefont {K.}~\bibnamefont {Sato}}, \bibinfo {author} {\bibfnamefont {D.}~\bibnamefont {Shiomi}},\ and\ \bibinfo {author} {\bibfnamefont {T.}~\bibnamefont {Takui}},\ }\bibfield  {title} {\bibinfo {title} {Adiabatic state preparation of correlated wave functions with nonlinear scheduling functions and broken-symmetry wave functions},\ }\href@noop {} {\bibfield  {journal} {\bibinfo  {journal} {Communications Chemistry}\ }\textbf {\bibinfo {volume} {5}},\ \bibinfo {pages} {84} (\bibinfo {year} {2022})}\BibitemShut {NoStop}%
\bibitem [{\citenamefont {Wiebe}\ \emph {et~al.}(2010)\citenamefont {Wiebe}, \citenamefont {Berry}, \citenamefont {H{\o}yer},\ and\ \citenamefont {Sanders}}]{wiebe2010higher}%
  \BibitemOpen
  \bibfield  {author} {\bibinfo {author} {\bibfnamefont {N.}~\bibnamefont {Wiebe}}, \bibinfo {author} {\bibfnamefont {D.}~\bibnamefont {Berry}}, \bibinfo {author} {\bibfnamefont {P.}~\bibnamefont {H{\o}yer}},\ and\ \bibinfo {author} {\bibfnamefont {B.~C.}\ \bibnamefont {Sanders}},\ }\bibfield  {title} {\bibinfo {title} {Higher order decompositions of ordered operator exponentials},\ }\href@noop {} {\bibfield  {journal} {\bibinfo  {journal} {Journal of Physics A: Mathematical and Theoretical}\ }\textbf {\bibinfo {volume} {43}},\ \bibinfo {pages} {065203} (\bibinfo {year} {2010})}\BibitemShut {NoStop}%
\bibitem [{\citenamefont {Berry}\ \emph {et~al.}(2020)\citenamefont {Berry}, \citenamefont {Childs}, \citenamefont {Su}, \citenamefont {Wang},\ and\ \citenamefont {Wiebe}}]{berry2020time}%
  \BibitemOpen
  \bibfield  {author} {\bibinfo {author} {\bibfnamefont {D.~W.}\ \bibnamefont {Berry}}, \bibinfo {author} {\bibfnamefont {A.~M.}\ \bibnamefont {Childs}}, \bibinfo {author} {\bibfnamefont {Y.}~\bibnamefont {Su}}, \bibinfo {author} {\bibfnamefont {X.}~\bibnamefont {Wang}},\ and\ \bibinfo {author} {\bibfnamefont {N.}~\bibnamefont {Wiebe}},\ }\bibfield  {title} {\bibinfo {title} {Time-dependent hamiltonian simulation with $ l^1$-norm scaling},\ }\href@noop {} {\bibfield  {journal} {\bibinfo  {journal} {Quantum}\ }\textbf {\bibinfo {volume} {4}},\ \bibinfo {pages} {254} (\bibinfo {year} {2020})}\BibitemShut {NoStop}%
\bibitem [{\citenamefont {Hu}\ and\ \citenamefont {Wu}(2016)}]{hu2016optimizing}%
  \BibitemOpen
  \bibfield  {author} {\bibinfo {author} {\bibfnamefont {H.}~\bibnamefont {Hu}}\ and\ \bibinfo {author} {\bibfnamefont {B.}~\bibnamefont {Wu}},\ }\bibfield  {title} {\bibinfo {title} {Optimizing the quantum adiabatic algorithm},\ }\href@noop {} {\bibfield  {journal} {\bibinfo  {journal} {Physical Review A}\ }\textbf {\bibinfo {volume} {93}},\ \bibinfo {pages} {012345} (\bibinfo {year} {2016})}\BibitemShut {NoStop}%
\bibitem [{\citenamefont {Rezakhani}\ \emph {et~al.}(2010)\citenamefont {Rezakhani}, \citenamefont {Pimachev},\ and\ \citenamefont {Lidar}}]{rezakhani2010accuracy}%
  \BibitemOpen
  \bibfield  {author} {\bibinfo {author} {\bibfnamefont {A.}~\bibnamefont {Rezakhani}}, \bibinfo {author} {\bibfnamefont {A.}~\bibnamefont {Pimachev}},\ and\ \bibinfo {author} {\bibfnamefont {D.}~\bibnamefont {Lidar}},\ }\bibfield  {title} {\bibinfo {title} {Accuracy versus run time in an adiabatic quantum search},\ }\href@noop {} {\bibfield  {journal} {\bibinfo  {journal} {Physical Review A—Atomic, Molecular, and Optical Physics}\ }\textbf {\bibinfo {volume} {82}},\ \bibinfo {pages} {052305} (\bibinfo {year} {2010})}\BibitemShut {NoStop}%
\bibitem [{\citenamefont {Lin}\ and\ \citenamefont {Tong}(2020)}]{lin2020near}%
  \BibitemOpen
  \bibfield  {author} {\bibinfo {author} {\bibfnamefont {L.}~\bibnamefont {Lin}}\ and\ \bibinfo {author} {\bibfnamefont {Y.}~\bibnamefont {Tong}},\ }\bibfield  {title} {\bibinfo {title} {Near-optimal ground state preparation},\ }\href@noop {} {\bibfield  {journal} {\bibinfo  {journal} {Quantum}\ }\textbf {\bibinfo {volume} {4}},\ \bibinfo {pages} {372} (\bibinfo {year} {2020})}\BibitemShut {NoStop}%
\bibitem [{\citenamefont {Babbush}\ \emph {et~al.}(2014)\citenamefont {Babbush}, \citenamefont {Love},\ and\ \citenamefont {Aspuru-Guzik}}]{babbush2014adiabatic}%
  \BibitemOpen
  \bibfield  {author} {\bibinfo {author} {\bibfnamefont {R.}~\bibnamefont {Babbush}}, \bibinfo {author} {\bibfnamefont {P.~J.}\ \bibnamefont {Love}},\ and\ \bibinfo {author} {\bibfnamefont {A.}~\bibnamefont {Aspuru-Guzik}},\ }\bibfield  {title} {\bibinfo {title} {Adiabatic quantum simulation of quantum chemistry},\ }\href@noop {} {\bibfield  {journal} {\bibinfo  {journal} {Scientific reports}\ }\textbf {\bibinfo {volume} {4}},\ \bibinfo {pages} {6603} (\bibinfo {year} {2014})}\BibitemShut {NoStop}%
\bibitem [{\citenamefont {Costa}\ \emph {et~al.}(2022)\citenamefont {Costa}, \citenamefont {An}, \citenamefont {Sanders}, \citenamefont {Su}, \citenamefont {Babbush},\ and\ \citenamefont {Berry}}]{costa2022optimal}%
  \BibitemOpen
  \bibfield  {author} {\bibinfo {author} {\bibfnamefont {P.~C.}\ \bibnamefont {Costa}}, \bibinfo {author} {\bibfnamefont {D.}~\bibnamefont {An}}, \bibinfo {author} {\bibfnamefont {Y.~R.}\ \bibnamefont {Sanders}}, \bibinfo {author} {\bibfnamefont {Y.}~\bibnamefont {Su}}, \bibinfo {author} {\bibfnamefont {R.}~\bibnamefont {Babbush}},\ and\ \bibinfo {author} {\bibfnamefont {D.~W.}\ \bibnamefont {Berry}},\ }\bibfield  {title} {\bibinfo {title} {Optimal scaling quantum linear-systems solver via discrete adiabatic theorem},\ }\href@noop {} {\bibfield  {journal} {\bibinfo  {journal} {PRX quantum}\ }\textbf {\bibinfo {volume} {3}},\ \bibinfo {pages} {040303} (\bibinfo {year} {2022})}\BibitemShut {NoStop}%
\bibitem [{\citenamefont {Harrow}\ \emph {et~al.}(2009)\citenamefont {Harrow}, \citenamefont {Hassidim},\ and\ \citenamefont {Lloyd}}]{harrow2009quantum}%
  \BibitemOpen
  \bibfield  {author} {\bibinfo {author} {\bibfnamefont {A.~W.}\ \bibnamefont {Harrow}}, \bibinfo {author} {\bibfnamefont {A.}~\bibnamefont {Hassidim}},\ and\ \bibinfo {author} {\bibfnamefont {S.}~\bibnamefont {Lloyd}},\ }\bibfield  {title} {\bibinfo {title} {Quantum algorithm for linear systems of equations},\ }\href@noop {} {\bibfield  {journal} {\bibinfo  {journal} {Physical review letters}\ }\textbf {\bibinfo {volume} {103}},\ \bibinfo {pages} {150502} (\bibinfo {year} {2009})}\BibitemShut {NoStop}%
\bibitem [{\citenamefont {Fin{\v{z}}gar}\ \emph {et~al.}(2025)\citenamefont {Fin{\v{z}}gar}, \citenamefont {Notarnicola}, \citenamefont {Cain}, \citenamefont {Lukin},\ and\ \citenamefont {Sels}}]{finvzgar2025counterdiabatic}%
  \BibitemOpen
  \bibfield  {author} {\bibinfo {author} {\bibfnamefont {J.~R.}\ \bibnamefont {Fin{\v{z}}gar}}, \bibinfo {author} {\bibfnamefont {S.}~\bibnamefont {Notarnicola}}, \bibinfo {author} {\bibfnamefont {M.}~\bibnamefont {Cain}}, \bibinfo {author} {\bibfnamefont {M.~D.}\ \bibnamefont {Lukin}},\ and\ \bibinfo {author} {\bibfnamefont {D.}~\bibnamefont {Sels}},\ }\bibfield  {title} {\bibinfo {title} {Counterdiabatic driving with performance guarantees},\ }\href@noop {} {\bibfield  {journal} {\bibinfo  {journal} {arXiv preprint arXiv:2503.01958}\ } (\bibinfo {year} {2025})}\BibitemShut {NoStop}%
\bibitem [{\citenamefont {Wan}\ and\ \citenamefont {Kim}(2020)}]{wan2020fast}%
  \BibitemOpen
  \bibfield  {author} {\bibinfo {author} {\bibfnamefont {K.}~\bibnamefont {Wan}}\ and\ \bibinfo {author} {\bibfnamefont {I.~H.}\ \bibnamefont {Kim}},\ }\bibfield  {title} {\bibinfo {title} {Fast digital methods for adiabatic state preparation},\ }\href@noop {} {\bibfield  {journal} {\bibinfo  {journal} {arXiv preprint arXiv:2004.04164}\ } (\bibinfo {year} {2020})}\BibitemShut {NoStop}%
\bibitem [{\citenamefont {Hormozi}\ \emph {et~al.}(2017)\citenamefont {Hormozi}, \citenamefont {Brown}, \citenamefont {Carleo},\ and\ \citenamefont {Troyer}}]{hormozi2017nonstoquastic}%
  \BibitemOpen
  \bibfield  {author} {\bibinfo {author} {\bibfnamefont {L.}~\bibnamefont {Hormozi}}, \bibinfo {author} {\bibfnamefont {E.~W.}\ \bibnamefont {Brown}}, \bibinfo {author} {\bibfnamefont {G.}~\bibnamefont {Carleo}},\ and\ \bibinfo {author} {\bibfnamefont {M.}~\bibnamefont {Troyer}},\ }\bibfield  {title} {\bibinfo {title} {Nonstoquastic hamiltonians and quantum annealing of an ising spin glass},\ }\href@noop {} {\bibfield  {journal} {\bibinfo  {journal} {Physical review B}\ }\textbf {\bibinfo {volume} {95}},\ \bibinfo {pages} {184416} (\bibinfo {year} {2017})}\BibitemShut {NoStop}%
\bibitem [{\citenamefont {Tran}\ \emph {et~al.}(2020)\citenamefont {Tran}, \citenamefont {Chu}, \citenamefont {Su}, \citenamefont {Childs},\ and\ \citenamefont {Gorshkov}}]{tranDestructiveErrorInterference2020}%
  \BibitemOpen
  \bibfield  {author} {\bibinfo {author} {\bibfnamefont {M.~C.}\ \bibnamefont {Tran}}, \bibinfo {author} {\bibfnamefont {S.-K.}\ \bibnamefont {Chu}}, \bibinfo {author} {\bibfnamefont {Y.}~\bibnamefont {Su}}, \bibinfo {author} {\bibfnamefont {A.~M.}\ \bibnamefont {Childs}},\ and\ \bibinfo {author} {\bibfnamefont {A.~V.}\ \bibnamefont {Gorshkov}},\ }\bibfield  {title} {\bibinfo {title} {Destructive {{Error Interference}} in {{Product-Formula Lattice Simulation}}},\ }\href {https://doi.org/10.1103/PhysRevLett.124.220502} {\bibfield  {journal} {\bibinfo  {journal} {Phys. Rev. Lett.}\ }\textbf {\bibinfo {volume} {124}},\ \bibinfo {pages} {220502} (\bibinfo {year} {2020})},\ \Eprint {https://arxiv.org/abs/1912.11047} {arXiv:1912.11047} \BibitemShut {NoStop}%
\bibitem [{\citenamefont {Layden}(2022{\natexlab{b}})}]{laydenFirstOrderTrotterError2022}%
  \BibitemOpen
  \bibfield  {author} {\bibinfo {author} {\bibfnamefont {D.}~\bibnamefont {Layden}},\ }\bibfield  {title} {\bibinfo {title} {First-{{Order Trotter Error}} from a {{Second-Order Perspective}}},\ }\href {https://doi.org/10.1103/PhysRevLett.128.210501} {\bibfield  {journal} {\bibinfo  {journal} {Phys. Rev. Lett.}\ }\textbf {\bibinfo {volume} {128}},\ \bibinfo {pages} {210501} (\bibinfo {year} {2022}{\natexlab{b}})},\ \Eprint {https://arxiv.org/abs/2107.08032} {arXiv:2107.08032} \BibitemShut {NoStop}%
\bibitem [{\citenamefont {Yi}\ and\ \citenamefont {Crosson}(2021)}]{yiSpectralAnalysisProduct2021}%
  \BibitemOpen
  \bibfield  {author} {\bibinfo {author} {\bibfnamefont {C.}~\bibnamefont {Yi}}\ and\ \bibinfo {author} {\bibfnamefont {E.}~\bibnamefont {Crosson}},\ }\href {http://arxiv.org/abs/2102.12655} {\bibinfo {title} {Spectral {{Analysis}} of {{Product Formulas}} for {{Quantum Simulation}}}} (\bibinfo {year} {2021}),\ \Eprint {https://arxiv.org/abs/2102.12655} {arXiv:2102.12655 [quant-ph]} \BibitemShut {NoStop}%
\end{thebibliography}%

\widetext
\appendix

\newpage

\section{Details of the Theorems}
Throughout this paper, we will use dot ( $\dot{ }$ ) for derivative to $x=t/T$ and prime ( $'$ ) for derivative to $u(x)$. For this reason, $\dot{H}(x):=dH(x)/dx$ and $H'(x):=dH(x)/du(x)$. 
We denote $H_i$ as the initial Hamiltonian, $H_f$ as the final Hamiltonian. 
The ground state of $H_f$ is prepared from the ground state of $H_i$ by slowly varying the Hamiltonian from $H_i$ to $H_f$. The instantaneous Hamiltonian is $H(t/T)=[1-u(t/T)]H_i + u(t/T)H_f$. Here $u(x)$ is the scheduling function satisfying $\dot{u}(x)\geq 0$, $u(0)=0$ and $u(1)=1$. The evolution time $T>>g^{-2}_{\text{min}}$ should be large enough to ensure the process is adiabatic, where $g_{\text{min}}$ is the minimum eigenvalue gap between the ground state and the first excited state of $H(t/T)$. 

We denote the state at time $t$ as $\psi(t/T)$, the state at $t=T$ is:
   
\begin{equation}
    |\psi(1)\rangle = \mathcal{T}e^{-i\int_0^TH(t/T)dt}|\psi(0)\rangle.
\end{equation}
To implement adiabatic quantum computation~(AQC) on a circuit model, we can approximate the evolution with a sequence of unitary operators~\cite{van2001powerful}:
\begin{equation}\label{eq:time_evolution_SI}
\begin{aligned}
|\psi(1)\rangle &= \prod_{m=1}^{r}U(m/r)|\psi(0)\rangle,\\
\end{aligned}
\end{equation}
where $U(m/r):=\exp[-iH(m/r)\delta t]$ and $r=T/\delta t$ is the number of time steps. The infidelity of this state is defined as:
\begin{equation}
\mathcal{I} := 1 - |\langle\psi(1)|\phi_0(1)\rangle|^2=\sum_{i\neq0}|\gamma_i(1)|^2,
\end{equation}
where
$\gamma_i(m/r) := \langle\phi_i(m/r)|\psi(m/r)\rangle$is the overlap between the state $|\psi(m/r)\rangle$ and the instantaneous state $|\phi_i(m/r)\rangle$ of $H(m/r)$ with the $i$-th lowest eigenvalue $E_i(m/r)$. The energy gap is defined as $g_i(m/r) = E_i(m/r) - E_0(m/r)$, and the average gap is $\bar{g}_i(m/r) = m^{-1}\sum_{m'=1}^m g_i(m'/r)$.

To approximate the exact time evolution operator $U(m/r)$ at step $m$, we can use Hamiltonian simulation methods, like Trotterization, LCU and GQSP, etc. to implement $\widetilde{U}(m/r)$ satisfying $\widetilde{U}(x)=U(x)+U_{\text{res}}(x)\delta t^{k+1} + O(\delta t^{k+2})$. Notice that sometimes $\widetilde{U}$ may not be unitary in some cases like first order LCU where we use $\widetilde{U}=1-iH\delta t$ to approximate $U=\exp(-iH\delta t)$. For this reason, we introduce the normalization factor at time step $m$ as:

\begin{equation}
a(m/r) = \sqrt{\langle \phi_0(m/r)|\widetilde{U}^{\dagger}(m/r)\widetilde{U}(m/r)|\phi_0(m/r)\rangle}.
\end{equation}
To describe of the diagonal terms of  $\widetilde{U}$ under the basis of $\{\phi_i\}$, we define the effective energy as $\widetilde{E}_i(m/r) := -\arg(\langle \phi_i|\widetilde{U}|\phi_i\rangle)/\delta t$. The average effective energy is defined as $\bar{E}_i(m/r) = m^{-1}\sum_{m'=1}^m \widetilde{E}_i(m'/r)$. And the effective gap is defined as $\widetilde{\Delta}_i(m/r) = \widetilde{E}_i(m/r) - \widetilde{E}_0(m/r)$. The average effective gap is defined as $\bar{\Delta}_i(m/r) = m^{-1}\sum_{m'=1}^m \widetilde{\Delta}_i(m'/r)$.  

Since $H(x)$ is a function defined at the interval $[0,1]$, the variables $g_i(m/r)$, $U(m/r)$, $\widetilde{U}(m/r)$, $\widetilde{E}_i(m/r)$, $\bar{E}_i(m/r)$ and $\widetilde{\Delta}_i(m/r)$ can all be defined continuously at the interval $[0,1]$ consequently. Sequentially, we define a continues version of $\bar{\Delta}_i(x)$ as $\bar{\Delta}_i(x) = x^{-1}\int_0^x \widetilde{\Delta}_i(x')dx'$.

Moreover, we define the decay rate:
\begin{equation}
\Lambda_i(x)=\frac{1}{\delta t} \ln\frac{\sqrt{|\langle \phi_0(x)|\widetilde{U}^{\dagger}(x)\widetilde{U}(x)|\phi_0(x)\rangle|}}{|\langle \phi_i|\widetilde{U}|\phi_i\rangle|}.
\end{equation}
We also define the average decay rate $\bar{\Lambda}_i(m/r) = m^{-1}\sum_{m'=1}^m \Lambda_i(m'/r)$ and its continuous version $\bar{\Lambda}_i(x)=x^{-1}\int_0^x \Lambda_i(x')dx'$.

We now proceed to prove some lemmas that will lead to the proof of the theorems.

\subsection{Proof of Lemma A.1}
%\textbf{Lemma A.1}:
\begin{lemma}
Suppose that a complex-valued function $\phi(x)$ is smooth in an open interval $[0,1]$, satisfying the following condition:

\begin{equation}\label{eq:condition_lemma}
    \left\{
    \begin{aligned}
        &\exp{(\lambda\Re(\phi(x)))}=O(1), ~~for ~~ \forall x\in [0,1],\\
         &\Im(\phi'(x))>0, ~~for ~~ \forall x\in [0,1],
    \end{aligned}
    \right.
\end{equation}
and $\psi(x)$ is a smooth real-valued function. With $\lambda\to +\infty$, for any $l\in \mathbb{N}^*$, we have:

\begin{equation}
\int_0^1 e^{\lambda\phi(x)}\psi(x)dx = \sum_{k=1}^{l}\frac{(-1)^{k-1}}{\lambda^k}e^{\lambda\phi(x)}\mathcal{L}^{k-1}_{\phi}\frac{\psi(x)}{i\phi'(x)}\Big|_0^1+O(\frac{1}{\lambda^{l+1}}),
\end{equation}
where $\mathcal{L}_{\phi}:= \frac{1}{i\phi'(x)}\frac{d}{dx}$.
\end{lemma}

\begin{proof}
    Using integration by parts, we have
\begin{equation}
\begin{aligned}
\int_0^1 e^{\lambda\phi(x)}\psi(x)dx &=\int_0^1 \psi(x)\frac{de^{i\lambda \phi(x)}}{i\lambda \phi'(x)}=\frac{\psi(x)}{i\lambda \phi'(x)}e^{\lambda\phi(x)}\Big|_0^1 - \int_0^1e^{i\lambda \phi(x)}dx\frac{d}{dx}\frac{\psi(x)}{i\lambda \phi'(x)}.\\
\end{aligned}
\end{equation}
Repeating this process for $l+1$ times, we obtain
\begin{equation}
\begin{aligned}
\int_0^1 e^{\lambda\phi(x)}\psi(x)dx &= \sum_{k=1}^{l+1}\frac{(-1)^{k-1}}{\lambda^k}e^{\lambda\phi(x)}\mathcal{L}^{k-1}_{\phi}\frac{\psi(x)}{i\phi'(x)}\Big|_0^1 + \frac{(-1)^{l+1}}{\lambda^{l+1}}\int_0^1e^{\lambda\phi(x)}\frac{d}{dx}\mathcal{L}^{l-1}_{\phi}\frac{\psi(x)}{i\phi'(x)}dx\\
&=\sum_{k=1}^{l}\frac{(-1)^{k-1}}{\lambda^k}e^{\lambda\phi(x)}\mathcal{L}^{k-1}_{\phi}\frac{\psi(x)}{i\phi'(x)}\Big|_0^1+O(\frac{1}{\lambda^{l+1}}).
\end{aligned}
\end{equation}
\end{proof}

\subsection{Lemma A.2}
%\textbf{Lemma A.2} 
\begin{lemma}[Theory of Trotter Error with Commutator Scaling \cite{childs2021theory}]
Let $H=\sum_{\gamma=1}^{\Gamma}H_\gamma$ be an operator consisting of $\Gamma$ summands, and let $t\geq 0$, $H_\gamma$ are anti-Hermitian. Let $\mathcal{S}(t)$ be a pth-order $\Upsilon$-stage product formula. Define $\widetilde{\alpha}_\text{comm}=\sum_{\gamma_1,\gamma_2,\cdots,\gamma_{p+1}=1}^\Gamma\|[H_{\gamma_(p+1)},\cdots[H_{\gamma_2},H_{\gamma_1}]\cdots]\|$, where $\|\cdot\|$ is the spectral norm. Then, the additive error $\mathcal{A}(t)$ defined by $\mathcal{S}(t)=e^{tH}+\mathcal{A}(t)$, can be asymptotically bounded as:
\begin{equation}
\|\mathcal{A}(t)\|=O(\widetilde{\alpha}_\text{comm}t^{p+1}).
\end{equation}
\end{lemma}
The proof of this Lemma can be found in Ref.~\cite{childs2021theory}.

\subsection{Proof of Lemma A.3}
%\textbf{Lemma A.3}:
\begin{lemma}
If $H_i$, $H_f$ are time-independent and non-degenerate, and the scheduling function $u(x)$ is smooth, $\phi_i(x)$ is the instantaneous eigenstate of $H(x)$ with eigenvalue $E_i(x)$, then the following relation holds:

\begin{equation}
\langle\phi_i(x)|\dot{\phi}_0(x)\rangle =\frac{\langle\phi_i(x)|\dot{u}H'(x)|\phi_0(x)\rangle}{g_i(x)}.
\end{equation} 
\end{lemma}

\begin{proof}
We start from the instantaneous eigenvalue equation for the Hamiltonian $H(x)$, where $|\phi_0(x)\rangle$ is the instantaneous ground state with the corresponding eigenvalue $E_0(x)$:
\begin{equation}
    H(x)|\phi_0(x)\rangle = E_0(x)|\phi_0(x)\rangle.
\end{equation}
To investigate how the eigenstate varies with the parameter $x$, we differentiate both sides of the above equation with respect to $x$:
\begin{equation}
    \frac{d H(x)}{dx}|\phi_0(x)\rangle + H(x)\frac{d}{dx}|\phi_0(x)\rangle = \frac{d E_0(x)}{dx}|\phi_0(x)\rangle + E_0(x)\frac{d}{dx}|\phi_0(x)\rangle.
\end{equation}
To solve for the transition matrix elements between different eigenstates, we left-multiply by the bra of another eigenstate, $\langle\phi_i(x)|$, where $i \neq 0$. By using the hermiticity of the Hamiltonian, such that $\langle\phi_i(x)|H(x) = E_i(x)\langle\phi_i(x)|$, and noting that the energy derivative $dE_0/dx$ is a scalar, we can use the orthogonality of the eigenstates, $\langle\phi_i(x)|\phi_0(x)\rangle = 0$, to eliminate a term. The equation then simplifies to:
\begin{equation}
    \langle\phi_i(x)|\frac{d H(x)}{dx}|\phi_0(x)\rangle + E_i(x)\langle\phi_i(x)|\frac{d}{dx}|\phi_0(x)\rangle = E_0(x)\langle\phi_i(x)|\frac{d}{dx}|\phi_0(x)\rangle.
\end{equation}
We can now rearrange this equation to solve for the term of interest, $\langle\phi_i(x)|\dot{\phi}_0(x)\rangle$. Considering $dH/dx$ can be expressed as $\dot{u}H'$, we finally obtain the well-known result:
\begin{equation}
    \langle\phi_i(x)|\dot{\phi}_0(x)\rangle = \frac{\langle\phi_i(x)|\dot{u}H'|\phi_0(x)\rangle}{E_0(x)-E_i(x)} = \frac{\langle\phi_i(x)|\dot{u}H'|\phi_0(x)\rangle}{g_i(x)}.
\end{equation}
\end{proof}

\subsection{Proof of Lemma A.4}\label{sec:details_of_GQSP}

\begin{lemma}[GQSP for Hamiltonian simulation~\cite{motlagh2024generalized}]
Let the Hamiltonian be given by a linear combination of unitaries $H=\sum_{j=1}^L\alpha_jU_j$ and $\alpha = \sum_{j=1}^L |\alpha_j|$. Assume access to a $\text{PREPARE}$ oracle that prepares the state $\alpha^{-1/2}\sum_{j=1}^L\sqrt{\alpha_j/\alpha}|j\rangle$, and a $\text{SELECT}$ oracle that applies the unitaries $U_j$ controlled by the state $|j\rangle$.
There exists a procedure based on $K$-th order GQSP that constructs an operator $\widetilde{U}(t)$ approximating the time-evolution operator $e^{-iHt}$  with probability of faliure less than $(1-\epsilon)^2$ using $O(L)$ anxiliary qubits and $O(K)$ queries to $\text{PREPARE}$ and $\text{SELECT}$ satisfying
$$
\| \widetilde{U}(t) - e^{-iHt}\| \leq \epsilon,
$$
and $K$ should scale as $O\left(\alpha t + \frac{\log(1/\epsilon)}{\log\log(1/\epsilon)}\right)$.
Here we define PREPARE: $|0\rangle \to \alpha^{-1}\sum_{j=1}^L\sqrt{\alpha_j}|j\rangle$ and SELECT: $|j\rangle|\psi\rangle \to |j\rangle U_j|\psi\rangle$.
\end{lemma}

\begin{proof}
To implement the time evolution operator $e^{-iHt}$ using GQSP, we first construct a unitary operator $W$ as follows:
\begin{equation}
W=-(1-2\text{PREPARE}|0\rangle\langle 0|\text{PREPARE}^\dagger)\text{SELECT}.
\end{equation}
Then we find that the span of $\text{PREPARE}|0\rangle| E_j\rangle$ and  $W\cdot\text{PREPARE}|0\rangle| E_j\rangle$ forms a two-dimensional invariant subspace of $W$. Here $| E_j\rangle$ is the eigenvector of $H$ with eigenvalue $ E_j$. Using the basis of $\text{PREPARE}|0\rangle| E_j\rangle$ and the orthogonal vector, the action of $W$ in the subspace can be represented as a $2\times 2$ matrix:
\begin{equation}
W_{ E_j}=
\left [
\begin{matrix}
 E_j/\alpha & \sqrt{1- E_j^2/\alpha^2}\\
-\sqrt{1- E_j^2/\alpha^2}& E_j/\alpha
\end{matrix}
\right ]
\end{equation}

Note that the eigenvalues of $W_{ E_j}$ are $e^{\pm i\arccos( E_j/\alpha)}$. Denote the corresponding eigenvectors as $|\phi_j^\pm\rangle$. Then we have

\begin{equation}
W=\sum_{j=1}^ne^{ i\arccos( E_j/\alpha)}|\phi_j^+\rangle\langle\phi_j^+|+e^{- i\arccos( E_j/\alpha)}|\phi_j^-\rangle\langle\phi_j^-|.
\end{equation}
Considering the formula:
\begin{equation}
e^{-it\cos\theta}=\sum_{k=-\infty}^\infty (-i)^kJ_k(t)e^{ik\theta},
\end{equation}
If we choose $P(x)=\sum_{k=-\infty}^\infty (-i)^kJ_k(\alpha t)x^k$, we have:
\begin{equation}
\begin{aligned}
P(W)&=\sum_{j=1}^ne^{ -i E_j t}|\phi_j^+\rangle\langle\phi_j^+|+e^{-i E_j| t}\phi_j^-\rangle\langle\phi_j^-|
=\bigoplus_{j=1}^n e^{-i E_j t} \mathbb{I}_{ E_j}
\end{aligned}
\end{equation}
Here $\mathbb{I}_{ E_j}$ is the identity operator in the subspace spanned by $\text{PREPARE}|0\rangle| E_j\rangle$ and  $W\cdot\text{PREPARE}|0\rangle| E_j\rangle$.
Therefore, we can get the desired time evolution operator by acting $P(W)$ on the state $\text{PREPARE}|0\rangle|\psi\rangle$:
\begin{equation}
\begin{aligned}
P(W)\text{PREPARE}|0\rangle|\psi\rangle&=\sum_{j=1}^n e^{-i E_j t}\text{PREPARE}|0\rangle| E_j\rangle\langle E_j|\psi\rangle
=\text{PREPARE}|0\rangle e^{-iHt}|\psi\rangle.
\end{aligned}
\end{equation}

To implement $P(W)$, we can use GQSP method. 
First, we need to truncate $P(x)$ to a polynomial of degree $K$ in terms of $x$ and $x^{-1}$:
\begin{equation}
P(x)\approx \sum_{k=-K}^K (-i)^kJ_k(\alpha t)x^k:=P_K(x).
\end{equation}
Applying Corollary 8 in \cite{motlagh2024generalized} we can implement the block encoding of $P_K(W)$ using GQSP with $O(K)$ applications of $W$ and  $O(\log L)$ ancilla qubits:
\begin{equation}
U_{\text{b-e}}(W)=
\left [
\begin{matrix}
P_K(W) & \cdot\\
\cdot& \cdot
\end{matrix}
\right ].
\end{equation}
By acting $U_{\text{b-e}}(W)$ on $|0\rangle\otimes\text{PREPARE}|0\rangle|\psi\rangle$, we can get the state:

\begin{equation}
\left [
\begin{matrix}
P_K(W) & \cdot\\
\cdot& \cdot
\end{matrix}
\right ]\cdot \left [
\begin{matrix}
\text{PREPARE}|0\rangle|\psi\rangle\\
0
\end{matrix}
\right ]=\left [
\begin{matrix}
P_K(W)\cdot\text{PREPARE}|0\rangle|\psi\rangle\\
\cdot
\end{matrix}
\right ].
\end{equation}
Then we can measure the first register to get the approximation of the state $\text{PREPARE}|0\rangle e^{-iHt}|\psi\rangle$. The probability of success depends on the norm of $P_K(W)\cdot\text{PREPARE}|0\rangle|\psi\rangle$.
To analyse the error and probability of success, using  $P_K(e^{i\theta})=P_K(e^{-i\theta})$, we have:

\begin{equation}
\begin{aligned}
P_K(W)&=\sum_{j=1}^ne^{ -i\lambda'_j t}|\phi_j^+\rangle\langle\phi_j^+|+e^{-i\lambda'_j t}|\phi_j^-\rangle\langle\phi_j^-|=\bigoplus_{j=1}^n e^{-i\lambda'_j t} \mathbb{I}_{ E_j}.
\end{aligned}
\end{equation}
Here $e^{-i\lambda'_j t}:=P_K(e^{\pm i\arccos(E_j/\alpha)})$. Note that $P_K(W)$ may not be unitary for $\lambda'_j$ may not be real. The imaginary part of $\lambda'_j$ will give rise to the probablity leakage.

\begin{equation}
\begin{aligned}
P_K(W)\text{PREPARE}|0\rangle|E_j\rangle&=\text{PREPARE}|0\rangle e^{-i\lambda'_jt}|E_j\rangle.\\
\end{aligned}
\end{equation}
Here we have
\begin{equation}
\begin{aligned}
e^{-i\lambda'_j t}&=P_K(e^{\pm i\arccos(E_j/\alpha)})=\sum_{k=-K}^K (-i)^kJ_k(\alpha t)e^{\pm i k\arccos(E_j/\alpha)}=J_0(\alpha t) + 2\sum_{k=1}^K (-i)^kJ_k(\alpha t)\cos(k\arccos(E_j/\alpha)).\\
\end{aligned}
\end{equation}
Then we can deduce that for GQSP, the effective time evolution operator is:
\begin{equation}
    \widetilde{U}(x)=P_K \left (e^{i\arccos(H(x)/\alpha)} \right ).
\end{equation}
The error from truncation is:
\begin{equation}
\begin{aligned}
e^{-i\lambda_j t}-e^{-i\lambda'_j t}&= 2\sum_{k=K+1}^{+\infty} (-i)^kJ_k(\alpha t)\cos(k\arccos(E_j/\alpha)).\\
\end{aligned}
\end{equation}
The truncation error can be bounded by the asymptotic behavior of Bessel functions:
\begin{equation}
\begin{aligned}
J_k(x)&=\sum_{m=0}^\infty \frac{(-1)^m}{m!(m+k)!}\Big(\frac{x}{2}\Big)^{2m+k}\subset O((\frac{ex}{2k})^k)\\
\end{aligned}
\end{equation}
From Ref.~\cite{motlagh2024generalized} we know that if we choose $K=K_{\epsilon}:=O(\alpha t + \log(1/\epsilon)/\log\log(1/\epsilon))$, then we can ensure that
\begin{equation}\label{eq:bound_GQSP}
\begin{aligned}
|e^{-i\lambda_j t}-e^{-i\lambda'_j t}|&\leq \epsilon.\\
\end{aligned}
\end{equation}
Therefore, the probability of success can be bounded as:
\begin{equation}
\begin{aligned}
\Big|P_K(W)\cdot\text{PREPARE}|0\rangle|E_j\rangle\Big|^2\geq (1-\epsilon)^2.\\
\end{aligned}
\end{equation}

If we use this method at each $r$ steps with time interval $\delta t$, the probability of success is $(1-\epsilon)^{2r}$. To keep the the probability of failure less than $\delta$, we need to choose $\epsilon=\delta/2r$. The order of GQSP is $K_{\epsilon}=O(\alpha \delta t + \log(r/\delta)/\log\log(r/\delta))$.
\end{proof}

\subsection{Proof of Lemma A.5}\label{sec:condition_of_GQSP}

\begin{lemma}
Using GQSP to implement each time evolution operator $U(m/r)$ at step $m$ in AQC, while $T\to +\infty$, the condition in Eq.~(\ref{eq:conditionA}) is satisfied if $\delta t$ is small enough and $K=O(\alpha\delta t+\log( r/\delta)/\log\log (r/\delta))$, where $K$ is the order of GQSP.

\end{lemma}

\begin{proof}
    First, we know that while $\delta t\to 0$ and $T \to +\infty $, $\widetilde{\Delta}_i(m/r)\to \Delta_i(m/r)>0$. So for sufficiently small $\delta t$, the latter condition is satisfied.
Next, from Eq.~(\ref{eq:bound_GQSP}), we know that
\begin{equation}
\|\widetilde{U}(m/r)-U(m/r)\|=O(\epsilon)=O(\frac{\delta}{r}).
\end{equation}
This result helps us to estimate the decay rate $\Lambda_i(m/r)$:
\begin{equation}
\begin{aligned}
\Lambda_i(m/r)&=\frac{1}{\delta t}\ln\frac{\sqrt{|\langle \phi_0(m/r)|\widetilde{U}^{\dagger}(m/r)\widetilde{U}(m/r)|\phi_0(m/r)\rangle|}}{|\langle \phi_i(m/r)|\widetilde{U}(m/r)|\phi_i(m/r)\rangle|}\\
&=\frac{1}{2\delta t}\ln|\langle \phi_0(m/r)|\widetilde{U}^{\dagger}(m/r)\widetilde{U}(m/r)|\phi_0(m/r)\rangle|-\frac{1}{\delta t}\ln|\langle \phi_i(m/r)|\widetilde{U}(m/r)|\phi_i(m/r)\rangle|\\
&=O(\frac{1}{T}).
\end{aligned}
\end{equation}
Then we can ensure that the first condition in Eq.~(\ref{eq:conditionA}) is satisfied:
\begin{equation}
e^{-\overline{\Lambda}_i(m/r)T}=\exp\Big(-\frac{1}{T}\sum_{m'=1}^m \Lambda_i(m'/r)\Big)=O(1).
\end{equation}
\end{proof}

\subsection{Proof of Theorem 1}\label{sec:proof_of_theorem1}

We are now ready to prove the theorems we presented in the main text:

% \noindent\textbf{{Theorem 1}}: 
\begin{theorem}[Theorem 1 in the main text]
    Assume that in a AQC process, at step $m$, we use $\widetilde{U}(m/r)$ to approximate the exact time evolution operator $U(m/r)$, where $\widetilde{U}(x)=U(x)+U_{\text{res}}(x)\delta t^{k+1} + O(\delta t^{k+2})$. While $\delta t$ is sufficiently small and  $T \to +\infty$, if the following condition
\begin{equation}\label{eq:conditionA}
    \left\{
    \begin{aligned}
        &\exp{(-\overline{\Lambda}_i(x)T)}=O(1),~~\forall x\in [0,1]\\
         &~~\widetilde{\Delta}_i(x)>0, ~~\forall x\in [0,1],
    \end{aligned}
    \right.
\end{equation}
is satisfied, the infidelity $\mathcal{I}$ scales as:
\begin{equation}\label{eq:bound1_SI}
    \mathcal{I}=O(\beta_{\text{ad}}^2\frac{1}{T^2}+ \beta_{\text{sim}}^2\delta t^{2k})
\end{equation}
Here $\beta_{\text{ad}}$ and $\beta_{\text{sim}}$ are coefficients for non-adiabatic error and quantum simulation error, which are irrelevant to $T$ and $\delta t$:
\begin{equation}
\begin{aligned}
&\beta_{\text{ad}}^2=\sum_{j\neq 0}\Big[|\frac{\dot{u}(1) \langle\phi_i(1)|H'(1)|\phi_0(1)\rangle }{g_{j}(1)^2}|^2+|\frac{\dot{u}(0) \langle\phi_i(0)|H'(0)|\phi_0(0)\rangle }{g_{j}(0)^2}|^2\Big],\\
&\beta_{\text{sim}}^2=\sum_{j\neq 0}|\frac{\langle\phi_i(1)|U_{\text{res}}(1)|\phi_0(1)\rangle }{g_{j}(1)}|^2+|\frac{ \langle\phi_i(0)|U_{\text{res}}(0)|\phi_0(0)\rangle }{g_{j}(0)}|^2.\\
\end{aligned}
\end{equation}
\end{theorem}

\begin{proof}
    The transition amplitude between the state $|\psi(\frac{m}{r})\rangle$ and the instantaneous state $|\phi_i(\frac{m}{r})\rangle$ ($i\neq 0$) of $H(\frac{m}{r})$ is:
\begin{equation}
\begin{aligned}
\gamma_i(\frac{m}{r})&=\langle\phi_i(\frac{m}{r})|\psi(\frac{m}{r})\rangle\\
&=\langle\phi_i(\frac{m}{r})|\frac{\widetilde{U}(m/r)}{a(m/r)}|\psi(\frac{m-1}{r})\rangle\\
&=  \sum_{j,k} \langle\phi_i(\frac{m}{r})|\frac{\widetilde{U}(m/r)}{a(m/r)}|\phi_j(\frac{m}{r})\rangle\langle\phi_j(\frac{m}{r})|\phi_k(\frac{m-1}{r})\rangle\gamma_k(\frac{m-1}{r}).\\
\end{aligned}
\end{equation}
Each term in the summation can be approximated as follows:
\begin{equation}
\begin{aligned}
\langle\phi_i(\frac{m}{r})|\frac{\widetilde{U}(m/r)}{a(m/r)}|\phi_i(\frac{m}{r})\rangle&=e^{-\Lambda_i(m/r)\delta t-i\widetilde{E}_i(m/r)\delta t}=1+O(\delta t),\\
\langle\phi_i(\frac{m}{r})|\widetilde{U}(m/r)|\phi_j(\frac{m}{r})\rangle&=\langle\phi_i(\frac{m}{r})|U_{\text{res}}(m/r)|\phi_j(\frac{m}{r})\rangle\delta t^{k+1}+O(\delta t^{k+2}), \text{ for $i\neq j$},\\
\langle\phi_j(\frac{m}{r})|\phi_k(\frac{m-1}{r})\rangle&=\delta_{j,k}+\frac{1}{r}\langle\dot{\phi_j}(m/r)|\phi_k(m/r)\rangle+O(\frac{1}{r^2}).\\
\end{aligned}
\end{equation}
Referring to the first-order disturbance theory, we further assume that $\gamma_0(m/r) \gg \gamma_i(m/r)$ and only consider the transition between the $i$th level and the ground state. Therefore, we obtain the following iteration equation:

\begin{equation}
\begin{aligned}
 \gamma_i(\frac{m}{r})&=\frac{1}{a(m/r)}\langle \phi_i|\widetilde{U}|\phi_i\rangle\gamma_i(\frac{m-1}{r})+\delta tR_i(m/r)\gamma_0(\frac{m-1}{r})+O(\delta t^{k+2}),\\
\end{aligned}
\end{equation}
where we define the transition amplitude from the ground state to the $i$th level as:
\begin{equation}
\begin{aligned}
R_i(m/r) &=R_i^{\text{ad}}(m/r)+\frac{1}{a(m/r)\delta t }\langle\phi_i(\frac{m}{r})|\widetilde{U}(m/r)|\phi_0(\frac{m}{r})\rangle\\
&= \frac{1}{T}\frac{\dot{u}(m/r)\langle \phi_i(m/r)|H'(m/r)|\phi_0(m/r)\rangle}{g_i(m/r)}+\frac{1}{a(m/r)}\langle \phi_i|U_{\text{res}}|\phi_0\rangle\delta t^k+O(\frac{1}{rT})+O(\delta t^{k+1}),\\
\end{aligned}
\end{equation}
and the adiabatic transition amplitude is defined as:
\begin{equation}
\begin{aligned}
R_{i}^{\text{ad}}(m/r):&=\frac{1}{\delta t}\langle\phi_i(m/r)|\phi_0(m-1/r)\rangle\\
&=\frac{1}{T}\langle\dot{\phi_i}(m/r)|\phi_0(m/r)\rangle+O(\frac{1}{rT})\\
&=\frac{1}{T}\frac{\dot{u}(m/r)\langle \phi_i(m/r)|H'(m/r)|\phi_0(m/r)\rangle}{g_{i}(m/r)}+O(\frac{1}{rT}).
\end{aligned}
\end{equation}
Here Lemma A.3 is used.

For simplicity, we define the following notations:
\begin{equation}
\begin{aligned}
C_i(m/r) &= \prod_{m'=1}^m a(m'/r)\langle \phi_i|\widetilde{U}|\phi_i\rangle^{-1}
= e^{[\bar{\Lambda}_i(m/r)+i\bar{E}_i(m/r)]Tm/r}.\\
\end{aligned}
\end{equation}
Multiplied $C_i(m/r)$ on both sides of the iteration equation above, we have:
\begin{equation}
C_i(\frac{m}{r})\gamma_i(\frac{m}{r})=C_i(\frac{m-1}{r})\gamma_i(\frac{m-1}{r})+R_i(\frac{m}{r})C_i(\frac{m}{r})\gamma_0(\frac{m-1}{r})\delta t.
\end{equation}
Summing over $m$ from $1$ to $m$, we get:
\begin{equation}
C_i(\frac{m}{r})\gamma_i(\frac{m}{r})=\sum_{m'=1}^{m-1}R_i(\frac{m'}{r})C_i(\frac{m'}{r})\gamma_0(\frac{m'}{r})\delta t.
\end{equation}

For first order approximation, we have:
\begin{equation}
\gamma_0(\frac{m}{r})=e^{-i\bar{E}_0(m/r)m\delta t}.
\end{equation}
For continuous version, we have the following expression:
\begin{equation}
\gamma_0(x)=e^{-i\bar{E}_0(x)xT}.
\end{equation}
Setting $m=r$ and changing the summation to integral, we get
\begin{equation}
\begin{aligned}
\gamma_i(1)&=Te^{-i\bar{E}_i(1)T}\int_0^1R_i(x)e^{[\bar{\Lambda}_i(x)x-\bar{\Lambda}_i(1)+i\bar{\Delta}_i(x)x]T}\left (1+O(\delta t)\right ) dx.
\end{aligned}
\end{equation}

The $(1+O(\delta t))$ term comes from the error of substituting a sum with an integral. To solve the integral, we define $\phi(x)=\bar{\Lambda}_i(x)x-\bar{\Lambda}_i(1)+i\bar{\Delta}_i(x)x$. Note that the derivative of $\phi(x)$ is:
\begin{equation}
\phi_i'(x)=\Lambda_i(x)+i\widetilde{\Delta}_{i}(x).
\end{equation}
From the condition in Eq.~(\ref{eq:conditionA}) we can infer that the condition Eq.~(\ref{eq:condition_lemma}) is satisfied, so we can use Lemma A.1 to get the formula below:

\begin{equation}\label{eq:cal_int}
\begin{aligned}
\gamma_i(1)&=\frac{R_i(1)e^{i\bar{\Delta}_i(1)T}}{\Lambda_i(1)+i\widetilde{\Delta}_i(1)}-\frac{R_i(0)e^{-\bar{\Lambda}_i(1)T}}{\Lambda_i(0)+i\widetilde{\Delta}_i(0)} +O(\frac{1}{T^2})\\
&\approx\frac{R_i(1)e^{i\bar{\Delta}_i(1)T}}{\Lambda_i(1)+i\widetilde{\Delta}_i(1)}-\frac{R_i(0)}{\Lambda_i(0)+i\widetilde{\Delta}_i(0)}e^{-\bar{\Lambda}_i(1)T}\\
&\approx\frac{R_i(1)e^{i\bar{g}_i(1)T}}{ig_i(1)}-\frac{R_i(0)}{ig_i(0)}e^{-\bar{\Lambda}_i(1)T}.\\
\end{aligned}
\end{equation}
Here we ignore the global phase factor $e^{-i\bar{E}_i(1)T}$. Considering that $\exp{(-\overline{\Lambda}_i(1)T)}=O(1)$, we can prove the infidelity scales as:
\begin{equation}
\mathcal{I}=\sum_{i\neq0}|\gamma_i|^2=O(\beta_{\text{ad}}^2\frac{1}{T^2}+ \beta_{\text{sim}}^2\delta t^{2k}).
\end{equation}
The expression of the infidelity upper bound might lead readers to a misunderstanding that as long as the energy gap of the initial Hamiltonian and the final Hamiltonian is large enough, the infidelity can be quite small through the whole evolution process despite that the energy gap in the middle is exponentially small. However, this is not the case. We derive the upper bound by expanding the infidelity around $T=+\infty$ as the asymptotic series and truncating at a certain order. The truncation is not a good approximation until the adiabatic condition $T>>g_{\text{min}}^{-2}$ is satisfied.
\end{proof}

\subsection{Proof of Theorem 2}\label{sec:proof_of_theorem2}

\begin{theorem}[Theorem 2 in the main text]
While $\delta t$ is sufficiently small and  $T \to +\infty$, the infidelity of the state prepared by the AQC process with the $k$-th order sub-Trotterization scales as $O(\beta_{\text{tro}}^2\delta t^{2k}+\beta_{\text{ad}}^2T^{-2})$. 
    
Here the coefficient $\beta_{\text{ad}}$ and $\beta_{\text{tro}}$ can be expressed as:
\begin{equation}
\begin{aligned}
&\beta_{\text{tro}}^2=\Big[|\frac{\widetilde{\alpha}^{(i)}_\text{comm}}{g_{1}(0)}|^2+|\frac{\widetilde{\alpha}^{(f)}_\text{comm}}{g_{1}(1)}|^2\Big],\\
&\beta_{\text{ad}}^2=\sum_{j\neq 0}\Big[|\frac{\dot{u}(1) \langle\phi_i(1)|H'(1)|\phi_0(1)\rangle }{g_{j}(1)^2}|^2+|\frac{\dot{u}(0) \langle\phi_i(0)|H'(0)|\phi_0(0)\rangle }{g_{j}(0)^2}|^2\Big].\\
\end{aligned}
\end{equation}
The commutator scaling factors $\widetilde{\alpha}^{(i)}_\text{comm}$ and $\widetilde{\alpha}^{(f)}_\text{comm}$ depend on the initial Hamiltonian $H_i$ and the final Hamiltonian $H_f$ respectively:
\begin{equation}
\begin{aligned}
&\widetilde{\alpha}^{(i)}_\text{comm}=\sum_{\gamma_1,\gamma_2,\cdots,\gamma_{k+1}=1}^{\Gamma_{i}}\prod_{t=1}^{k+1}c^{(i)}_{\gamma_t}\|[P^{(i)}_{\gamma_{k+1}},\cdots[P^{(i)}_{\gamma_2},P^{(i)}_{\gamma_1}]\cdots]\|,\\
&\widetilde{\alpha}^{(f)}_\text{comm}=\sum_{\gamma_1,\gamma_2,\cdots,\gamma_{k+1}=1}^{\Gamma_{f}}\prod_{t=1}^{k+1}c^{(f)}_{\gamma_t}\|[P^{(f)}_{\gamma_{k+1}},\cdots[P^{(f)}_{\gamma_2},P^{(f)}_{\gamma_1}]\cdots]\|.\\
\end{aligned}
\end{equation}
Here we assume that the initial Hamiltonian $H_i$ and the final Hamiltonian $H_f$ are both decomposed into the sum of $\Gamma_{i/f}$ Pauli operators:
\begin{equation}
H_i=\sum_{\gamma=1}^{\Gamma_i}c^{(i)}_\gamma P^{(i)}_{\gamma},~~~H_f=\sum_{\gamma=1}^{\Gamma_f}c^{(f)}_\gamma P^{(f)}_{\gamma}.
\end{equation}
\end{theorem}

\begin{proof}
    We denote the $k$-th order trotter formula of $e^{-iH_{i/f}[1-u(m/r)]\delta t}$ as $U_{i/f}^k(m/r)$. The additional error introduced by the k-th order trotter formula is $\mathcal{A}_{i/f}^k(u)=O(t^{k+1})$. Moreover, we will use $\mathcal{A}^k(m/r)$ to represent the additional error introduced by the $k$-th order primary Trotterization formula of $\exp(-i\{[1-u(m/r)]H_i+u(m/r)H_f\}\delta t)$. That is:

\begin{equation}
\begin{aligned}
e^{-i\{[1-u(m/r)]H_i+u(m/r)H_f\}\delta t} &= e^{-iH_i(1-u(m/r))\delta t}e^{-iH_f(u(m/r))\delta t} + \mathcal{A}^k(m/r),\\
e^{-iH_i(1-u(m/r))\delta t} &= U_{i}^k(m/r) + \mathcal{A}_i^k(m/r),\\
e^{-iH_f(u(m/r))\delta t} &= U_{f}^k(m/r) + \mathcal{A}_f^k(m/r).
\end{aligned}
\end{equation}
So we can express $\widetilde{U}(m/r)$ as the sum of the exact time evolution operator and the additional errors:
\begin{equation}
\begin{aligned}
\widetilde{U}(\frac{m}{r}) &= U_i^{k}(\frac{m}{r})U_f^{k}(\frac{m}{r})\\
&=  (e^{-iH_iu(\frac{m}{r})\delta t}+\mathcal{A}_i^k(\frac{m}{r})) (e^{-iH_f[1-u(\frac{m}{r})]\delta t}+\mathcal{A}_f^k(\frac{m}{r}))\\
&\approx  [e^{-iH_iu(\frac{m}{r})\delta t}e^{-iH_f[1-u(\frac{m}{r})]\delta t}+\mathcal{A}_i^k(\frac{m}{r})e^{-iH_f[1-u(\frac{m}{r})]\delta t}+e^{-iH_iu(\frac{m}{r})\delta t}\mathcal{A}_f^k(\frac{m}{r}) ]\\
&\approx  e^{-iH(\frac{m}{r})\delta t} + \mathcal{A}^k(\frac{m}{r}) + \mathcal{A}^k_i(\frac{m}{r}) + \mathcal{A}^k_f(\frac{m}{r}).
\end{aligned}
\end{equation}

As we defined above in the proof of Theorem 1,
\begin{equation}
R_i(m/r) :=  R_i^{\text{ad}}(m/r)+\frac{1}{a(m/r)\delta t}\langle \phi_i|\widetilde{U}|\phi_0\rangle.
\end{equation}
Notice that as $\widetilde{U}$ is unitary, the normalization factor $a(m/r)=1$. Furthermore, we break $R_i(m/r)$ into four parts:

\begin{equation}
R_i(m/r) = R_i^{\text{ad}}(m/r)+R_i^{\text{pri}}(m/r)+R_i^{\text{sub,i}}(m/r)+R_i^{\text{sub,f}}(m/r),
\end{equation}
each of which is defined as follows:
\begin{equation}
\begin{aligned}
R_i^{\text{ad}}(m/r)&\approx\frac{1}{T}\frac{\dot{u}(m/r)\langle \phi_i(m/r)|H'(m/r)|\phi_0(m/r)\rangle}{g_i(m/r)},\\
R_i^{\text{pri}}(m/r)&:=\frac{1}{\delta t}\langle \phi_i|\mathcal{A}^k(m/r)|\phi_0\rangle,\\
R_i^{\text{sub,i}}(m/r)&:= \frac{1}{\delta t}\langle \phi_i|\mathcal{A}^k_i(m/r)|\phi_0\rangle,\\
R_i^{\text{sub,f}}(m/r)&:=\frac{1}{\delta t}\langle \phi_i|\mathcal{A}^k_f(m/r)|\phi_0\rangle.
\end{aligned}
\end{equation}

After deriving the expression of $R_i(m/r)$, we move to validate the condition in Eq.~(\ref{eq:conditionA}). Note that for Trotterization, $\widetilde{U}$ is unitary. So we have:

\begin{equation}
\begin{aligned}
\Lambda_i(m/r) := \frac{1}{\delta t} \ln\left(\frac{\sqrt{|\langle \phi_0(m/r)|\widetilde{U}^{\dagger}(m/r)\widetilde{U}(m/r)|\phi_0(m/r)\rangle|}}{|\langle \phi_i(m/r)|\widetilde{U}(m/r)|\phi_i(m/r)\rangle|}\right)>0.
\end{aligned}
\end{equation}
Thus the first condition is satisfied. For sufficiently small $\delta t$, the second condition is also satisfied. So the condition in Eq.~(\ref{eq:conditionA}) is satisfied. Using Eq.~(\ref{eq:cal_int}) in the proof of Theorem 1, we get
\begin{equation}
\begin{aligned}
\gamma_i(1)&\approx\frac{R_i(1)e^{i\bar{g}_i(1)T}}{ig_i(1)}-\frac{R_i(0)}{ig_i(0)}e^{-\bar{\Lambda}_i(1)T}
= \gamma_i^{\text{ad}}+\gamma_i^{\text{pri}}+\gamma_i^{\text{sub,i}}+\gamma_i^{\text{sub,f}},
\end{aligned}
\end{equation}
where $\gamma_i^{\text{ad}}$, $\gamma_i^{\text{pri}}$, $\gamma_i^{\text{sub,i}}$, and $\gamma_i^{\text{sub,f}}$ are the contributions from $R_i^{\text{ad}}$, $R_i^{\text{pri}}$, $R_i^{\text{sub,i}}$, and $R_i^{\text{sub,f}}$ respectively. 

Using Theorem 1 and Lemma A.3, we get the following expression of the adiabatic term:
\begin{equation}
\begin{aligned}
\gamma_i^{\text{ad}}&=\frac{R_i^{\text{ad}}(1)e^{iT\bar{g}(1)}}{ig_i(1)}-\frac{R_i^{\text{ad}}(0)}{ig_i(0)}e^{-\bar{\Lambda}_i(1)T},\\
|\gamma_i^{\text{ad}}|&\leq\frac{1}{T}\Big[|\frac{ \langle\dot{\phi_i}(1)|\phi_0(1)\rangle }{g_i(1)}|+|\frac{ \langle\dot{\phi_i}(0)|\phi_0(0)\rangle }{g_i(0)}|\Big]\\
&=\frac{1}{T}\Big[|\frac{\dot{u}(1) \langle\phi_i(1)|H_i|\phi_0(1)\rangle}{g^2_i(1)}|+|\frac{ \dot{u}(0)\langle\phi_i(0)|H_f|\phi_0(0)\rangle}{g^2_i(0)}|\Big]\\
&=O(T^{-1}).
\end{aligned}
\end{equation}

For the sub-Trotterization terms, using Theorem 1, we get 
\begin{equation}
\begin{aligned}
\gamma_i^{\text{sub,i}} &=\frac{R_i^{\text{sub,i}}(1)e^{iT\bar{g}(1)}}{ig_i(1)}-\frac{R_i^{\text{sub,i}}(0)}{ig_i(0)}+O(T^{-1})=-\frac{R_i^{\text{sub,i}}(0)}{ig_i(0)}+O(T^{-1}).\\
\end{aligned}
\end{equation}
Similarly, we also have:
\begin{equation}
\gamma_i^{\text{sub,f}}=\frac{R_i^{\text{sub,f}}(1)}{ig_i(1)}+O(T^{-1}).
\end{equation}
As for the primary Trotterization term, note that $R_i^{\text{pri}}(0)=R_i^{\text{pri}}(1)=0$, the leading term of $\gamma_i^{\text{pri}}$ will vanish. So we need to consider the next order term:
\begin{equation}
\gamma_i^{\text{pri}}=O(T^{-1}\delta t^k).
\end{equation}

To bound the sub-Trotterization terms, we need to estimate the norm of $\mathcal{A}^k_i(x)$ and $\mathcal{A}^k_f(x)$. From Lemma A.2, we know that:
\begin{equation}
\begin{aligned}
&\|\mathcal{A}^k_i(x)\|=O\left ([1-u(x)]^{k+1}\widetilde{\alpha}^{(i)}_\text{comm}\delta t^{k+1}\right ),\\
&\|\mathcal{A}^k_f(x)\|=O\left (u(x)^{k+1}\widetilde{\alpha}^{(f)}_\text{comm}(x)\delta t^{k+1}\right ).\\
\end{aligned}
\end{equation}
Therefore, we finally get the following bound of the infidelity:
\begin{equation}
\begin{aligned}
\mathcal{I} &\leq \sum_{i\neq0} \Big[|\gamma_i^{\text{ad}}|^2+|\gamma_i^{\text{sub,i}}|^2+|\gamma_i^{\text{sub,f}}|^2\Big]\\
&=\beta_{\text{ad}}^2T^{-2}+\sum_{i\neq0} \Big[\Big|\frac{\langle \phi_i(1)|\mathcal{A}^k_i(1)|\phi_0(1)\rangle}{\delta tg_i(1)}\Big|^2+\Big|\frac{\langle \phi_i(0)|\mathcal{A}^k_f(0)|\phi_0(0)\rangle}{\delta tg_i(1)}\Big|^2 \Big]\\
&\leq \beta_{\text{ad}}^2T^{-2}+\Big[\Big|\frac{\|\mathcal{A}^k_i(1)\|}{\delta tg_1(1)}\Big|^2+\Big|\frac{\|\mathcal{A}^k_f(0)\|}{\delta tg_1(0)}\Big|^2\Big]\\
&=\beta_{\text{ad}}^2T^{-2}+\beta_{\text{tro}}^2\delta t^{2k}.
\end{aligned}
\end{equation}
So the infidelity of the state prepared by the AQC process with the k-th order sub-Trotterization scales as $O(\beta_{\text{tro}}^2\delta t^{2k}+\beta_{\text{ad}}^2T^{-2})$. Here $\beta_{\text{tro}}$ and $\beta_{\text{ad}}$ are constants depending on the $H_i$, $H_f$ and the scheduling function $u(x)$.

\end{proof}

\subsection{Proof of Corollary 2.1}\label{sec:proof_of_corollary2.1}

\begin{corollary}[Corollary 2.1 in the main text, Optimization of the AQC process under fixed circuit depth]
While $\delta t$ is sufficiently small and  $T \to +\infty$, parameters $\delta t, T, k$ of the AQC process with sub-Trotterization can be optimized when the circuit depth $d$ is fixed by minimizing the infidelity of the state prepared by the AQC process with the $k$-th order sub-Trotterization.
The optimal parameters and the upper bound of the infidelity are given as follow. 

For $k=1$:
\begin{equation}
\begin{aligned}
&\mathcal{I}_{opt}=O(2\beta_{\text{tro}}\beta_{\text{ad}}Dd^{-1}),\\
&T_{opt}=\sqrt{\frac{1}{D}\frac{\beta_{\text{ad}}}{\beta_{\text{tro}}}}d^{\frac{1}{2}},\\
&\delta t_{opt}=\sqrt{D\frac{\beta_{\text{ad}}}{\beta_{\text{tro}}}}d^{-\frac{1}{2}}.
\end{aligned}
\end{equation}
For $k\geq 2$ and $k$ being an even number:
\begin{equation}
\begin{aligned}
&\mathcal{I}_{opt}=O(\frac{k+1}{k^{\frac{k}{k+1}}}\sqrt[k+1]{(2D\cdot 5^{k/2-1})^{2k}\beta_{\text{tro}}^2\beta_{\text{ad}}^2}\cdot d^{-\frac{2k}{k+1}}),\\
&T_{opt}= \sqrt[k+1]{\frac{k+1}{\sqrt{k}(2D\cdot 5^{k/2}-1)^k}\frac{\beta_{\text{ad}}}{\beta_{\text{tro}}}}\cdot d^{\frac{k}{k+1}},\\
&\delta t_{opt}=\sqrt[k+1]{\frac{k+1}{\sqrt{k}}(2D\cdot 5^{k/2}-1)\frac{\beta_{\text{ad}}}{\beta_{\text{tro}}}}\cdot d^{-\frac{1}{k+1}}.
\end{aligned}
\end{equation}
Here we define:
\begin{equation}
D := \Gamma_i+\Gamma_f.
\end{equation}
\end{corollary}

\begin{proof}
    For $k=1$, the depth of the circuit is $d= D r$. Here for simplicity, we ignore the constant factor in the depth of the circuit. The upper bound of the infidelity is:
\begin{equation}
\begin{aligned}
\beta_{\text{tro}}^2\delta t^{2}+\beta_{\text{ad}}^2T^{-2}&=\beta^2_{tro}\frac{D^2}{d^2}T^2+\beta^2_{\text{ad}}T^{-2}\geq 2\beta_{\text{tro}}\beta_{\text{ad}}Dd^{-1}.\\
\end{aligned}
\end{equation}
The optimal parameters can be obtained at the minimum of the upper bound of the infidelity.

\begin{equation}
\begin{aligned}
&T_{opt}=\sqrt{\frac{1}{D}\frac{\beta_{\text{ad}}}{\beta_{\text{tro}}}}d^{\frac{1}{2}},\\
&\delta t_{opt}=\sqrt{D\frac{\beta_{\text{ad}}}{\beta_{\text{tro}}}}d^{-\frac{1}{2}}.\\
\end{aligned}
\end{equation}
For $k\geq 2$ and $k$ being an even number, the depth of the circuit is $d=2\cdot 5^{k/2-1} D r$. The upper bound of the infidelity is:

\begin{equation}
\begin{aligned}
\beta_{\text{tro}}^2\delta t^{2k}+\beta_{\text{ad}}^2T^{-2}
&=\left (\frac{2D\cdot 5^{k/2-1}}{d}\right )^{2k}\beta_{\text{tro}}^2 T^{2k}+\beta_{\text{ad}}^2T^{-2}\\
&\geq (k+1)\sqrt[k+1]{\left (\frac{2D\cdot 5^{k/2-1}}{\sqrt{k}}\right )^{2k}\beta_{\text{tro}}^2\beta_{\text{ad}}^2}\cdot d^{-\frac{2k}{k+1}}\\
&=\frac{k+1}{k^{\frac{k}{k+1}}}\sqrt[k+1]{\left (2D\cdot 5^{k/2-1}\right )^{2k}\beta_{\text{tro}}^2\beta_{\text{ad}}^2}\cdot d^{-\frac{2k}{k+1}}.\\
\end{aligned}
\end{equation}
The optimal parameters can be obtained at the minimum of the upper bound of the infidelity.

\begin{equation}
\begin{aligned}
T_{opt}&=\sqrt[k+1]{\frac{k+1}{\sqrt{k}(2D\cdot 5^{k/2}-1)^k}\frac{\beta_{\text{ad}}}{\beta_{\text{tro}}}}\cdot d^{\frac{k}{k+1}},\\
\delta t_{opt}&=\sqrt[k+1]{\frac{k+1}{\sqrt{k}}(2D\cdot 5^{k/2}-1)\frac{\beta_{\text{ad}}}{\beta_{\text{tro}}}}\cdot d^{-\frac{1}{k+1}}.\\
\end{aligned}
\end{equation}

Corollary 2.1 also indicates that if the desired circuit depth is small, low-order Trotterization is good enough for high accuracy. While the circuit depth is large, higher-order Trotterization is preferred.
\end{proof}

\subsection{Proof of Theorem 3}\label{sec:proof_of_theorem3}

\begin{theorem}[Theorem 3 in the main text, Strong error cancellation with GQSP]
{
While $\delta t$ is sufficiently small and  $T \to +\infty$, with the order $K$ scaling as $O(\log(r/\delta)/\log\log(r/\delta))$, GQSP can simulate digital adiabatic evolution using $O(\log (L))$ anxiliary qubits while the infidelity $\mathcal I(1)$ scales as $O(\beta_{\text{ad}}^2T^{-2})$, the probability of failure will be under $\delta$. Here $\beta_{\text{ad}}$ is independent of $T$ and $\delta t$.}
\end{theorem}

\begin{proof}
    Using Lemma A.5, we know that the condition in Eq.~(\ref{eq:conditionA}) is satisfied.
Using Theorem 1, we get the infidelity $\mathcal{I}$ scales as:
\begin{equation}
    \mathcal{I}=O(\beta_{\text{ad}}^2\frac{1}{T^2}+ \beta_{\text{sim}}^2\delta t^{2k}).
\end{equation}
Here $\beta_{\text{ad}}$ and $\beta_{\text{sim}}$ are coefficients for non-adiabatic error and Trotter error, which are irrelevant to $T$ and $\delta t$:
\begin{equation}
\begin{aligned}
\beta_{\text{ad}}^2&=\sum_{j\neq 0}\Big[|\frac{\dot{u}(1) \langle\phi_i(1)|H'(1)|\phi_0(1)\rangle }{g_{j}(1)^2}|^2+|\frac{\dot{u}(0) \langle\phi_i(0)|H'(0)|\phi_0(0)\rangle }{g_{j}(0)^2}|^2\Big],\\
\beta_{\text{sim}}^2&=\sum_{j\neq 0}\Big[\Big|\frac{\langle\phi_i(1)|U_{\text{res}}(1)|\phi_0(1)\rangle }{g_{j}(1)}\Big|^2+\Big|\frac{ \langle\phi_i(0)|U_{\text{res}}(0)|\phi_0(0)\rangle }{g_{j}(0)}\Big|^2\Big].\\
\end{aligned}
\end{equation}
Note that $U(x)$ and $\widetilde{U}(x)$ are all diagonal under the basis $\{\phi_i(x)\}$, so $\langle\phi_i(x)|U_{\text{res}}(x)|\phi_i(x)\rangle=0$ and $\beta_{\text{sim}}=0$. For this reason, the Infidelity $\mathcal{I} = \sum_{i\neq 0}|\gamma_i|^2=O(\beta_{\text{ad}}^2T^{-2})$. Furthermore, we can learn from Lemma A.4 that the probability of failure is under $\delta$ and $O(\log (L))$ anxiliary qubits are needed.
\end{proof}

\subsection{Proof of Theorem 4}\label{sec:proof_of_corollary3.1}

\begin{theorem}[Theorem 4 in the main text, Exponential error cancellation with $Q$-th order path]
{While $\delta t$ is sufficiently small and  $T \to +\infty$, the infidelity $\mathcal I(1)$ indigital adiabatic evolution using GQSP scales as $O(\beta^2_{\text{GQSP}}T^{-2Q-2})$ if we choose the scheduling function $u(x)$ as the $Q$-th order path. Moreover, if we choose an $\infty$-order path, the infidelity scales as $O(\beta^2_{\text{GQSP}}e^{-T})$. Here $\beta_{\text{GQSP}}$ is a coefficient that depends on $K$, the path, and the Hamiltonian, but independent of $T$. }
\end{theorem}

\begin{proof}
    From Theorem 1 we know that:
\begin{equation}
\begin{aligned}
\gamma_i(1)&=Te^{-i\bar{E}_i(1)T}\int_0^1R_i^{\text{ad}}(x)e^{[\bar{\Lambda}_i(x)x-\bar{\Lambda}_i(1)+i\bar{\Delta}_i(x)x]T} dx=\int_0^1\frac{\dot{u}\langle\phi_i(x)|H'|\phi_0(x)\rangle}{\Delta_{i}(x)}e^{[\bar{\Lambda}_i(x)x-\bar{\Lambda}_i(1)+i\bar{\Delta}_i(x)x]T} dx.
\end{aligned}
\end{equation}
Here we ignore the global phase factor $e^{-i\bar{E}_i(1)T}$. To simplify the notation, we define:

\begin{equation}
\begin{aligned}
&\psi(x):=\frac{\dot{u}\langle\phi_i(x)|H'|\phi_0(x)\rangle}{\Delta_{i}(x)},\\
&\phi(x):=\bar{\Lambda}_i(x)x-\bar{\Lambda}_i(1)+i\bar{\Delta}_i(x)x.
\end{aligned}
\end{equation}
Using Lemma A.1, we can expand $\gamma_i(1)$ as the asymptotic series of $T^{-1}$:

\begin{equation}
\begin{aligned}
\gamma_i(1) =
\int_0^1 e^{\phi(x)T}\psi(x)dx = \sum_{q=1}^{Q}\frac{(-1)^{q-1}}{T^q}e^{\phi(x)T}\mathcal{L}^{q-1}_{\phi}\frac{\psi(x)}{i\phi'(x)}\Big|_0^1+O(\frac{1}{T^{Q+1}}).\\
\end{aligned}
\end{equation}
Here, $\mathcal{L}_{\phi}:= \frac{1}{i\phi'(x)}\frac{d}{dx}$. From Definition 1 we know that

\begin{equation}
u^{(q)}(x)|_{x=0}=u^{(q)}(x)|_{x=0}=0, \quad q=1,2,\cdots, Q-1.
\end{equation}
Therefore, we have $\mathcal{L}^{q-1}_{\phi}\frac{\psi(x)}{i\phi'(x)}\Big|_0=\mathcal{L}^{q-1}_{\phi}\frac{\psi(x)}{i\phi'(x)}\Big|_1=0$, for $q=1,2,\cdots, Q-1$. So the first $Q-1$ terms in the asymptotic series vanish. Thus we have:

\begin{equation}
|\gamma_i| = O(\beta_{i}T^{-Q}).
\end{equation}
Here $\beta^2_{i}=\Big|\mathcal{L}^{Q-1}_{\phi}\frac{\psi(x)}{i\phi'(x)}|_{x=0}\Big|^2+\Big|\mathcal{L}^{Q-1}_{\phi}\frac{\psi(x)}{i\phi'(x)}|_{x=1}\Big|^2$. Then we can get the bound of the infidelity as:
\begin{equation}
    \mathcal{I} = \sum_{i\neq 0}|\gamma_i|^2=O(\beta^2_{\text{GQSP}}T^{-2Q}).
\end{equation}
Here $\beta_{\text{GQSP}}^2=\sum_i \beta_i^2$.
For an $\infty$ order path, the proof is similar to that in~\cite{lin2020near}, so we omit it.
\end{proof}

\section{How Theorem 2 Reduces to ``Self-Healing'' Theorem}
% Self-healing of AQC \cite{kovalsky2023self}:

While $\delta t$ is sufficiently small and  $T \to +\infty$, for the AQC process with only first-order primary Trotterization or the AQC process with first-order primary Trotterization and $k$-th order sub-Trotterization, the components of $H_i$ and $H_f$ commutes separately, the infidelity of the prepared state by scales as $O(\delta t^2T^{-2}+T^{-2})$.

\begin{proof}
    From the given condition we know that $R_i^{\text{sub,i}}=R_i^{\text{sub,f}}=0$, $R_i^{\text{pri}}=O(\delta t)$, and $R_i^{\text{pri}}(0)=R_i^{\text{pri}}(1)=0$. Using Lemma A.1 and Eq.~(\ref{eq:cal_int}) in the proof of Theorem 1  we can get:
\begin{equation}
\gamma_i^{\text{sub,i}}=\gamma_i^{\text{sub,f}}=0,~
\gamma_i^{\text{pri}}=O\left(\frac{1}{T}e^{iT[\bar{\Delta}(x)]x}\mathcal{L}_{\bar{\Delta}(x)x}\frac{R_i^{\text{pri}}(x)}{i\widetilde{\Delta}(x)}\Big|_0^1\right)=O(T^{-1}\delta t).
\end{equation}
Using Theorem 1 and Lemma A.3, we get the adiabatic term is:
\begin{equation}
\begin{aligned}
\gamma_i^{\text{ad}}&=\frac{R_i^{\text{ad}}(1)e^{iT\bar{\Delta}(1)}}{ig_i(1)}-\frac{R_i^{\text{ad}}(0)}{ig_i(0)}e^{-\bar{\Lambda}_i(1)T},\\
|\gamma_i^{\text{ad}}|&\leq\frac{1}{T}\Big[|\frac{ \langle\dot{\phi_i}(1)|\phi_0(1)\rangle }{g_i(1)}|+|\frac{ \langle\dot{\phi_i}(0)|\phi_0(0)\rangle }{g_i(0)}|\Big]\\
&=\frac{1}{T}\Big[|\frac{\dot{u}(1) \langle\phi_i(1)|H_i|\phi_0(1)\rangle}{g^2_i(1)}|+|\frac{ \dot{u}(0)\langle\phi_i(0)|H_f|\phi_0(0)\rangle}{g^2_i(0)}|\Big]\\
&=O(T^{-1}).
\end{aligned}
\end{equation}
Here, the Infidelity $\mathcal{I} = \sum_{i\neq 0}|\gamma_i|^2=O(\delta t^2T^{-2}+T^{-2})$.
\end{proof}

\end{document}